\documentclass[12pt,a4paper]{article}
\usepackage[margin=1.0in]{geometry}
\usepackage{times}
\usepackage{setspace}
\usepackage{graphicx}
\usepackage{amsmath,amssymb,amsthm}
\usepackage{hyperref}
\usepackage[ruled,vlined]{algorithm2e}
\usepackage{float}
\restylefloat{algorithm2e}
\usepackage{caption}
\usepackage{lipsum}

\def\etal{{et al.~}}

\newcommand{\remove}[1]{}

\newtheorem{theo}{Theorem}
\newtheorem{lem}[theo]{Lemma}
\newtheorem{corr}[theo]{Corollary}

\newenvironment{keyword}[1][Keywords]{\begin{trivlist}
\item[\hskip \labelsep {\bfseries #1}]}{\end{trivlist}}

\usepackage{color}

\let\geq\geqslant
\let\leq\leqslant

\let\le\leqslant

\makeatletter
\partopsep\z@ \textfloatsep 10pt plus 1pt minus 4pt
\def\section{\@startsection {section}{1}{\z@}{-3.0ex plus -0.5ex minus
-.2ex}{2.3ex plus .2ex}{\large\bf}}
\def\subsection{\@startsection{subsection}{2}{\z@}{-2.5ex plus -0.5ex
minus -.2ex}{1.5ex plus .2ex}{\normalsize\bf}}
\def\@fnsymbol#1{\ensuremath{\ifcase#1\or 1\or 2\or
    3\or 4\or 5\or 6\or 7 \or 8\ or 9 \or 10\or 11 \else\@ctrerr\fi}}
\makeatother

\title{An Algorithm for Computing Constrained Reflection Paths 
in Simple Polygons}
\author{Arijit Bishnu,\thanks{Advanced Computing and Microelectronics Unit,
Indian Statistical Institute, Kolkata 700108, India
$~~~~~~~~$(arijit@isical.ac.in).}
  \and Subir Kumar Ghosh,\thanks{School of Technology \& Computer Science, Tata
Institute of Fundamental Research, Mumbai 400005,
India $~~~~~~~~$(ghosh@tifr.res.in).}%
  \and Partha Pratim Goswami,\thanks{Institute of Radiophysics and Electronics,
University of Calcutta, Kolkata 700009,
India $~~~~~~~~$(ppg.rpe@caluniv.ac.in).}%
  \and Sudebkumar Prasant Pal,\thanks{Department of Computer Science and
Engineering, Indian Institute of Technology, Kharagpur 721302, India
$~~~~~~~~$(spp@cse.iitkgp.ernet.in).}%
  \and Swami Sarvattomananda\thanks{School of Mathematical Sciences,
Ramakrishna Mission Vivekananda University, Belur, India 711202,
$~~~~~~~$ (shreesh@rkmvu.ac.in).}%
}

\date{}

\begin{document}

\maketitle
\begin{abstract}
Let $s$ be a source point and $t$ be a destination point inside an $n$-vertex 
simple polygon $P$. Euclidean shortest paths [Lee and Preparata, Networks,
1984; Guibas \etal, Algorithmica, 1987] and minimum-link paths [Suri, CVGIP,
1986; Ghosh, J. Algorithms, 1991] between $s$ and $t$ inside $P$ have been well
studied. Both these kinds of paths are simple and piecewise-convex.
\remove{Similar paths in the context of diffuse or specular reflections have
not been studied.} 
However, computing optimal paths in the context of diffuse or specular
reflections does not seem to be an easy task.
A path from a light source $s$ to $t$ inside $P$ is called a
{\it diffuse reflection path} if the turning points of the path lie in the
interiors of the boundary edges of $P$. A diffuse reflection path is said to be
{\it optimal} if it has the minimum number of turning points amongst all diffuse
reflection paths between $s$ and $t$. The minimum diffuse reflection path may
not be simple. The problem of computing the minimum diffuse reflection path in
low degree polynomial time has remained open.

In our quest for understanding the geometric structure of the minimum diffuse
reflection paths vis-a-vis shortest paths and minimum link paths, we define a
new kind of diffuse reflection path called a {\it constrained diffuse
reflection path} where (i) the path is simple, (ii) it intersects only the
{\it eaves} of the Euclidean shortest path between $s$ and $t$, and 
(iii) it intersects each eave exactly once. For computing a minimum
constrained diffuse reflection path from $s$ to $t$, we present an
$O(n(n+\beta))$ time algorithm, where $\beta =\Theta (n^2)$  in the worst
case. Here, $\beta$ depends on the shape of the polygon. We also establish some
properties relating minimum constrained diffuse reflection paths and minimum
diffuse reflection paths. Constrained diffuse reflection paths introduced in
this paper provide new geometric insights into the hitherto unknown structures
and shapes of optimal reflection paths. Our algorithm demonstrates how
properties like convexity, simplicity, complete visibility, etc., can be
combined in computing and understanding diffuse reflection paths that are
optimal or close to optimal.
\end{abstract}

\newpage

\section{Introduction}
\label{sec:intro}
\remove{Let $s$ be a source point and $t$ be a destination point inside an
$n$-vertex simple polygon $P$. Reachability problems in terms of Euclidean
shortest paths \cite{ghlst-ltavs-87,lp-net-84} and minimum link paths
\cite{g-cvpcs-91,s-ltamlp-86} between a source point $s$ and a destination 
point $t$ inside a simple polygon $P$ have been well
studied. Both these kinds of paths are simple and have piecewise
convexity property. Similar paths in the context of diffuse or specular
reflections have not been studied. In this paper, we study such diffuse
reflection paths.}

\subsection{Visibility and reflections}
\label{ssec:vis-and-reflect} 
Problems of {\it direct} visibility have been studied 
extensively in the last few decades (see \cite{ghosh-book-2007}).
Let $P$ be an $n$-vertex simple polygon where $int(P)$ and $bd(P)$
denote the interior and boundary of $P$, respectively. 
Two points inside a polygon $P$ are said to be $visible$ (directly) if the line
segment joining them lies totally inside $int(P) \cup bd(P)$. The region of $P$
visible directly from a point light source $s$ in $P$ is called the 
{\it visibility polygon} of $P$ from $s$ (see Figure \ref{ldiffusespecular}). 
Efficient algorithms have been designed for computing visibility 
polygons under various conditions \cite{ghosh-book-2007}. 
Note that some points of $P$ that are not directly visible or illuminated from
$s$, can become visible due to multiple reflections on the edges of $P$ (see
Figure \ref{ldiffusespecular}). 

\medskip

We are interested in computing the visibility of a point from $s$ by multiple
reflections inside $P$; the sequence of multiple reflections is simply a
reflection path. Reflections are of two types -- specular and diffuse.
As per the law of reflection, the reflection of a light
ray at a point is called {\it specular} if the angle of incidence is equal to
the angle of reflection. The other type of reflection of light is
called {\it diffuse} reflection that happens for most reflecting surfaces.
Here, a ray incident 
at a point of an edge $e$ is 
reflected in all possible interior directions except along the edge $e$.
We assume that all edges of $P$ can reflect in this manner.
We also assume that any ray of light incident at a vertex is absorbed 
and not reflected.

\begin{figure*}[ht]
\begin{minipage}[b]{1.0\textwidth}
\centering
\includegraphics[width=0.6\columnwidth]{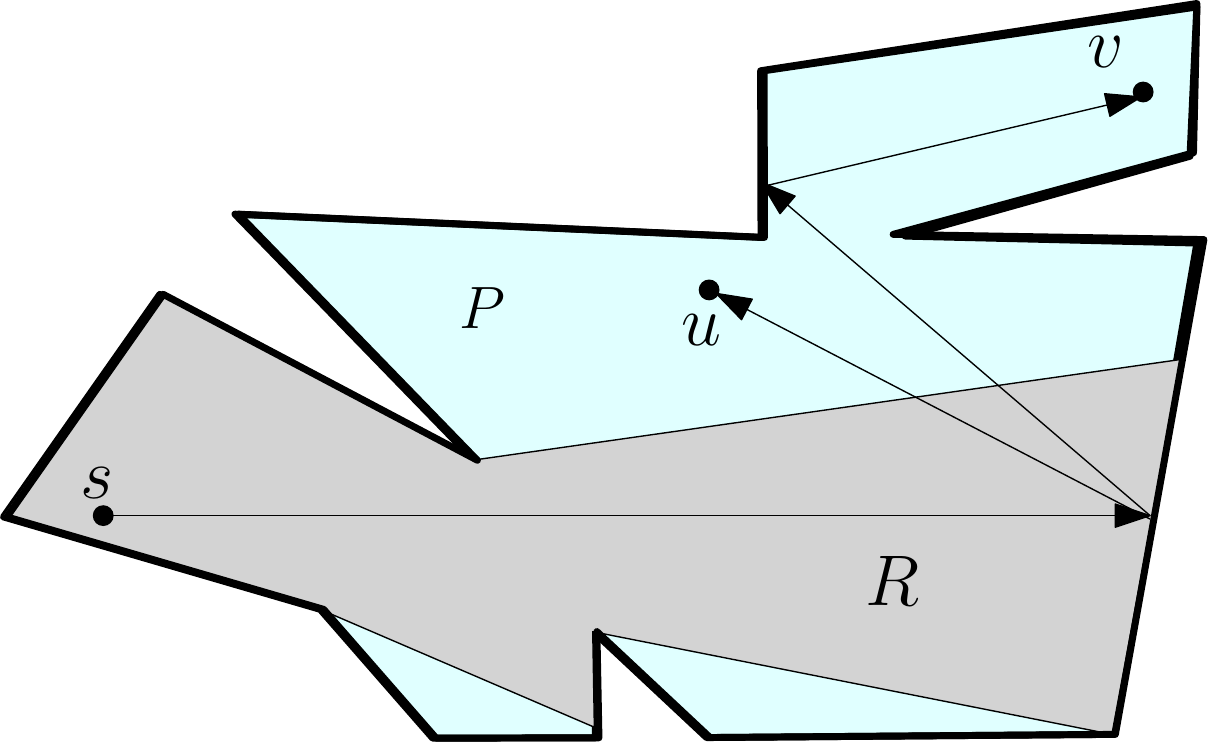}\\
\captionof{figure}{The region $R$ is directly visible from $s$. 
A ray from $s$ reaches $u$ after one specular 
reflection. A ray from $s$ reaches $v$ after two 
diffuse reflections.}
\label{ldiffusespecular}
\end{minipage}
\end{figure*}
\begin{figure*}[ht]
\begin{minipage}[b]{1.0\textwidth}
  \centering
 \includegraphics[width=0.7\columnwidth]{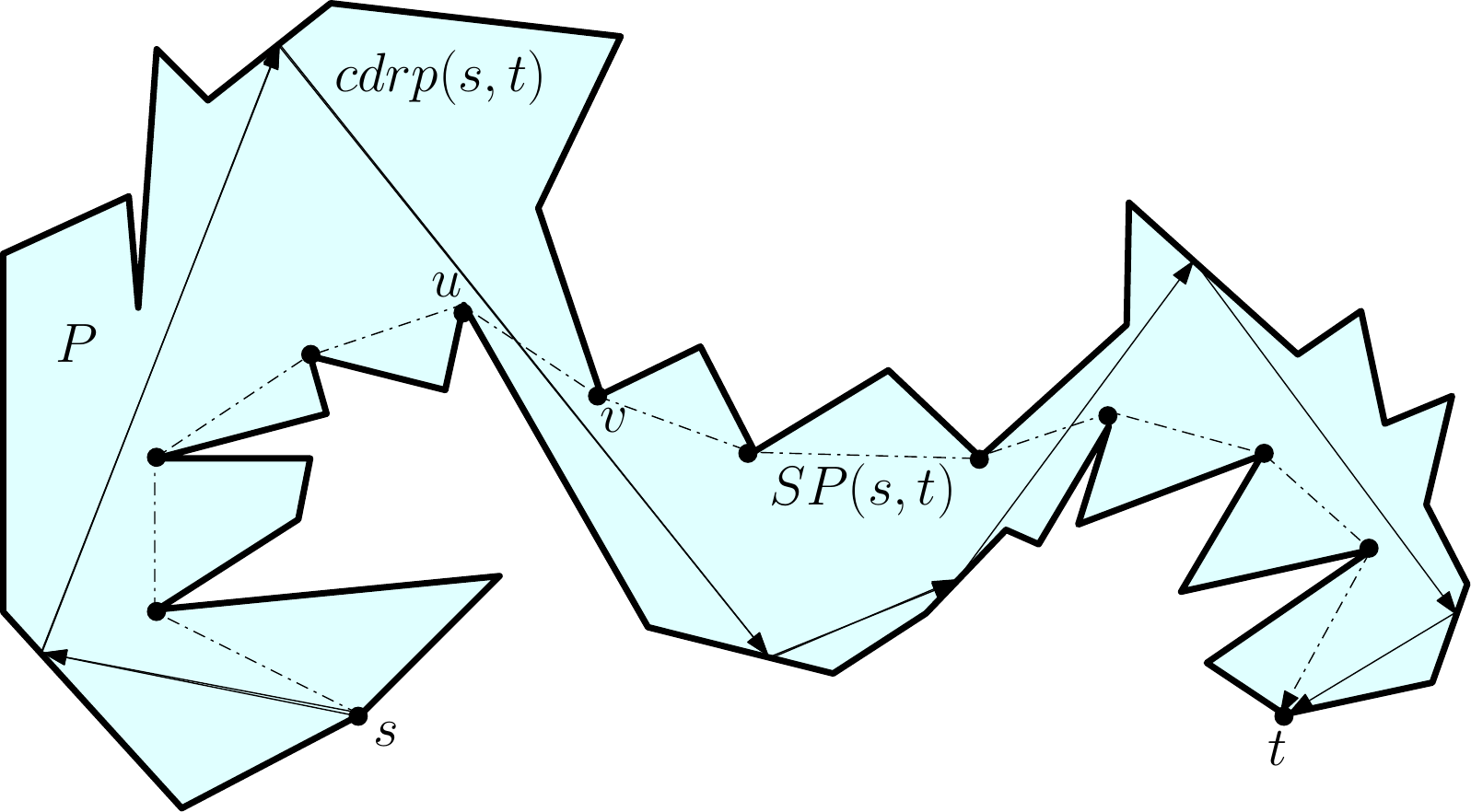}\\
\captionof{figure}{A $cdrp(s,t)$ intersects only eaves of $SP(s,t)$.}
\label{constrainedpath}
\end{minipage}
\end{figure*}

\medskip

Multiple reflections arise naturally in the realistic rendering of
three-dimensional scenes 
\cite{addpp-vmr-96,b-rsdo-1993,fdfhp-94,ggmnps-acdr-2012}.
In rendering of images by ray-tracing, light sources reachable by a 
small number of reflections through an image pixel would contribute
intensely at the pixel because of limited loss of intensity at each stage of
reflection. Computing diffuse reflection paths of light arriving 
from a light source by a small number of reflections (or turning points) 
is therefore, an important problem. 
In this paper we focus on the polynomial time computation of
certain constrained paths of multiple diffuse
reflections. Prior works on visibility with multiple reflections are reviewed in
Section~\ref{sec:prevresult}.

\subsection{Euclidean shortest path and minimum link path}
\label{ssec:short-link-path}
A polygonal path is said to be {\it simple} if it is not self intersecting. 
Henceforth, we use the term path instead of polygonal path. 
A path from $s$ to $t$ inside $P$ is said to be {\it convex} if it makes only
unidirectional turns (either left or right) while traversing from $s$ to $t$. 
If the turns are always right (or, left) turns, the path is right convex
(respectively, left convex). A path from $s$ to $t$ inside $P$ is said to be
{\it piecewise-convex} if it can be broken into alternating sequences of left
and right convex paths.

Let $SP(s,t)$ denote the Euclidean shortest path, the path of the shortest
length, \remove{\cite{ghlst-ltavs-87,lp-net-84}} from $s$ to $t$ inside $P$.
Let $uv$ be an edge of $SP(s,t)$ such that if $P$ is cut along $uv$, then $s$
and $t$ belong to the two different subpolygons of $P$ (see Figure
\ref{constrainedpath}). Such an edge $uv$ is called an {\it eave}
\cite{ghosh-book-2007}. Notice that $uv$ is a diagonal of $P$.
$SP(s,t)$ is simple and piecewise convex with reversal of turns at eaves.

\medskip

A minimum link path \cite{g-cvpcs-91,s-ltamlp-86} between two points $s$ 
and $t$ (denoted as $mlp(s,t)$) is a polygonal path inside $P$ having the
minimum number of turns. Like $SP(s,t)$, $mlp(s,t)$ is also simple and 
piecewise convex. Moreover, $mlp(s,t)$ intersects $SP(s,t)$ 
only at eaves and each eave is intersected exactly once \cite{ghosh-book-2007}.
The number of links in a minimum link path between any
two points of $P$ is called the {\it link distance} between them. 

\subsection{Diffuse reflection paths}
\label{ssec:reflectionpath}
As stated earlier, reachability problems in terms of Euclidean
shortest paths and minimum link paths between a source point $s$ 
and a destination point $t$ inside a simple polygon $P$ have been well
studied \cite{ghosh-book-2007}.
In this paper, we seek to compute a special class of optimal diffuse reflection
paths that are analogous to $SP(s,t)$ and $mlp(s,t)$. We call such paths 
\emph{constrained diffuse reflection paths} which we define next. 

\medskip

A diffuse reflection path $drp(s,t)$ from $s$ to $t$ is a path inside $P$ from
$s$ to $t$ such that the turns of
the path are in the interiors of the edges of $bd(P)$. Note that every
$drp(s,t)$ must intersect all the eaves of $SP(s,t)$. If all turning
points of a $mlp(s,t)$ lie on $bd(P)$, then $mlp(s,t)$ is a 
$drp(s,t)$. A $drp(s,t)$ is said to be {\it optimal} if it has the minimum
number of reflections amongst all diffuse reflection paths between $s$ and $t$.
An optimal $drp(s,t)$ can always be computed in exponential time
\cite{ppd-vmdr-98}. Aronov et al.~\cite{adiy-cdfs-2006} claimed that 
the combinatorial complexity of the visible region after $k$ diffuse 
reflections is $O(n^9)$, for any $k \le n$. It seems that this result
may be used for computing an $mdrp(s,t)$ in very high order polynomial 
time but no explicit procedure is stated in the paper. Designing a low degree 
polynomial time algorithm for computing an $mdrp(s,t)$ remains open.
In Section \ref{sec:prevresult}, we state the known approximation algorithms for
computing $drp(s,t)$.

\medskip

A $drp(s,t)$ is said to be {\it constrained} if (i) it is simple, (ii) it 
intersects only the {\it eaves} of $SP(s,t)$, and (iii) it intersects 
each eave exactly once. Such a path is denoted by $cdrp(s,t)$; see 
Figure~\ref{constrainedpath} for an illustration. 
If every $drp(s,t)$ intersects itself or intersects a non-eave edge of
$SP(s,t)$, then there cannot exist a $cdrp(s,t)$. 
In Figure~\ref{constrainedpath1}, any $cdrp(s,t)$ has to enter triangle 
$abc$ from $ab$ and exit from $bc$. Such a $cdrp(s,t)$ must end up on the 
clockwise polygonal boundary 
from $f$ to $g$ and therefore, can not reach $t$, where $f$ is the
extension of $ab$ to $bd(P)$ and $g$ is the extension of $bc$ to $bd(P)$. 

If a $cdrp(s,t)$ has the minimum number of turns amongst all $cdrp(s,t)$, then
it is denoted as $mcdrp(s,t)$. Ghosh \cite{g-cvpcs-91,ghosh-book-2007} showed
how $SP(s,t)$ can be transformed into  a $mlp(s,t)$  in $P$. Following a similar
approach as shown in the sequel, $SP(s,t)$ can also be transformed to a
$cdrp(s,t)$ (if it exists).
In addition, they all have the same number of reversals of turns from $s$ to $t$
only at eaves; the turns immediately before and after each link crossing an
eave have reverse directions, just as the turns at the two end of an eave in
$SP(s,t)$ have reverse directions.
Thus, one can observe that the particular diffuse reflection
path we are interested in has significant structural similarities with
the Euclidean shortest paths and the minimum link paths.

\begin{figure*}[ht]
\begin{minipage}[b]{1.0\textwidth}
\centering
\includegraphics[width=0.6\columnwidth]{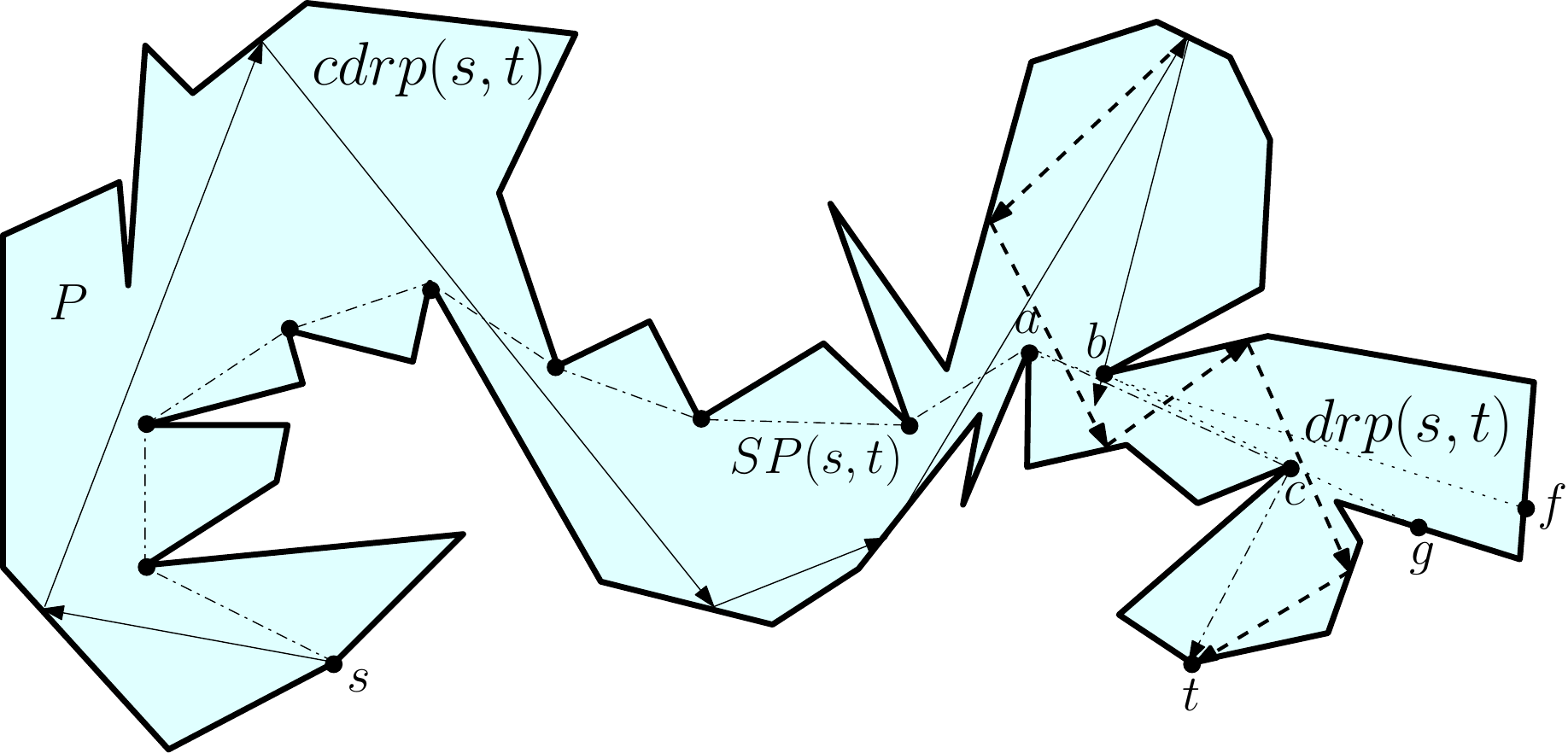}\\
\captionof{figure}{There is no $cdrp(s,t)$ as any path must 
intersect non-eave edge of $SP(s,t)$.}
\label{constrainedpath1}
\end{minipage}
\end{figure*}
\begin{figure*}[ht]
 \begin{minipage}[b]{1.0\textwidth}
  \centering
 \includegraphics[width=0.5\columnwidth]{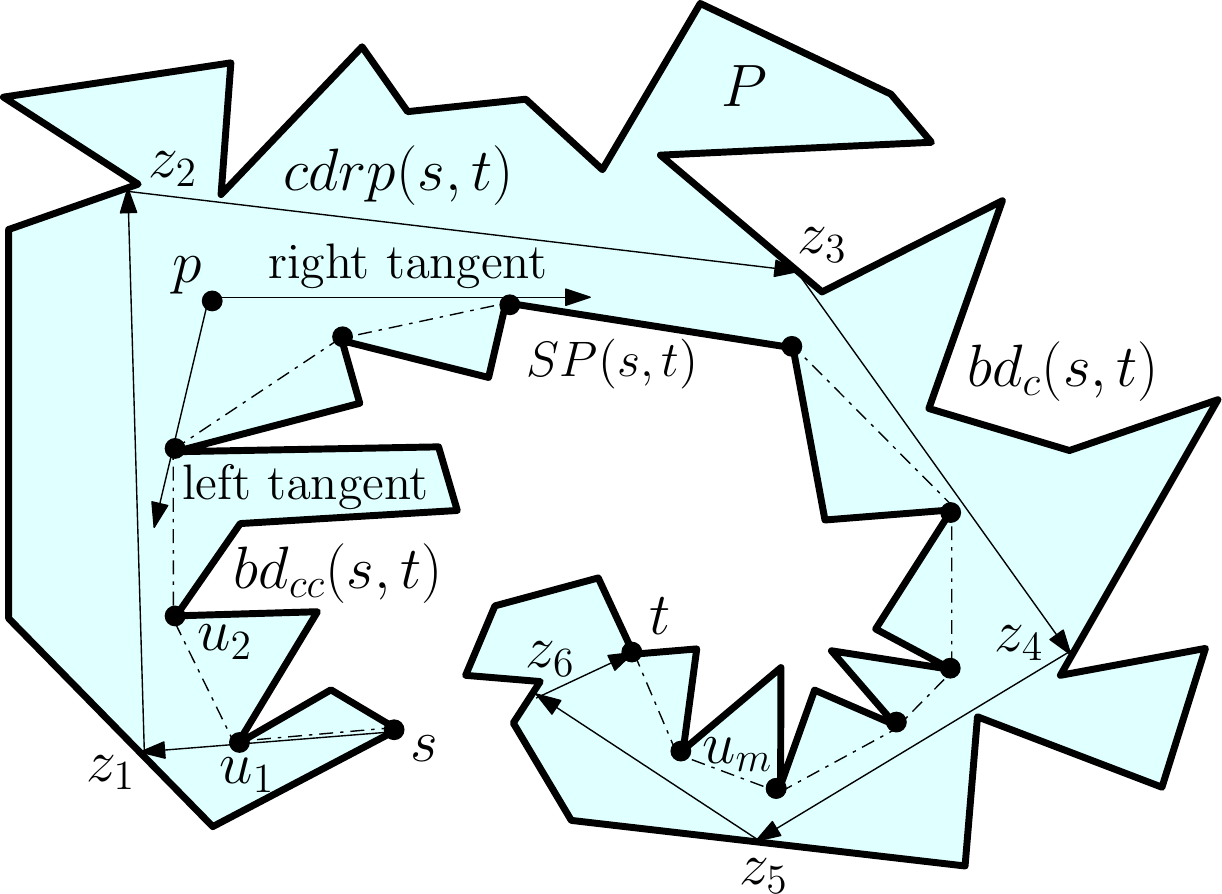}\\
\captionof{figure}{Every $cdrp(s,t)$ is a convex and simple path.}
\label{eaveno1}
\end{minipage}
\end{figure*}

\medskip
Section~\ref{sec:prevresult} reviews previous work on visibility 
with multiple reflections. 
The main goal of this paper is to compute an optimal constrained diffuse
reflection path $mcdrp(s,t)$ that is dealt with in Section
\ref{sec:mcdrp}. We present an $O(n(n+\beta))$ time algorithm for 
computing $mcdrp(s,t)$. Here, $\beta = \Theta (n^2)$ depends on the 
shape of the polygon. To the best of our knowledge, this is the first 
attempt at computing any class of optimal diffuse reflection path. 
Section~\ref{sec:explore} relates $mcdrp(s,t)$ with $drp(s,t)$ and 
Section~\ref{sec:conclude} concludes the paper.

\medskip

\section{Previous results}
\label{sec:prevresult}
We state some known results on visibility with multiple reflections.
In~\cite{addpp-vr-95}, Aronov et al.~studied the region visible from a point
source inside a simple $n$-vertex polygon where at most one specular 
(or diffuse) reflection is permitted on the bounding edges. They established 
a tight $\Theta(n^2)$ worst-case combinatorial complexity bound for 
the region visible after at most one reflection. They also proposed an
algorithm for computing such regions in $O(n^2\log^2 n)$ time for both 
specular as well as diffuse reflections. Aronov et al.~\cite{addpp-vmr-96}
addressed the general problem where at most $k\geq 2$ specular reflections are
used. An upper bound of $O(n^{2k})$ and a worst-case lower bound of
$\Omega((n/k)^{2k})$ was established on the combinatorial complexity of the
region visible due to at most a constant number $k$ of {\it specular}
reflections. They also proposed an algorithm running in $O(n^{2k}\log n)$ 
time, for $k>1$.
\remove{
}
Davis \cite{dr-vlt-1998} studied several variations of 
reflection problems.

\medskip

Prasad et al.~\cite{ppd-vmdr-98} showed that the upper bound on the 
number of edges and vertices of the region visible due to at most $k$
diffuse reflections is $O(n^{2\lceil (k+1)/2\rceil +1})$.
They designed an $O(n^{2\lceil (k+1)/2\rceil +1}\log n)$ time
algorithm for computing such a visible region. In~\cite{ppd-vmdr-98} 
they conjectured that the complexity of the 
region visible due to at most $k$ diffuse reflections is 
$\Theta (n^2)$. Note that this region may contain blind spots or holes (see
in \cite{pbs-lwlb-04}). Aronov et al.~\cite{adiy-cdfs-2006} 
claimed that the complexity of this visible region  
is $O(n^9)$. Bridging the big gap between the $O(n^9)$ upper bound 
of Aronov et al. \cite{adiy-cdfs-2006}, and
the $\Omega (n^2)$ lower bound as in~\cite{ppd-vmdr-98}, is an 
open problem.
 
\medskip

Recently, Ghosh et al.~\cite{ggmnps-acdr-2012} have presented
three different algorithms for computing sub-optimal diffuse reflection
paths from $s$ to $t$ inside $P$. For constructing such a path, 
the first algorithm uses a greedy method, the second algorithm uses a 
transformation of  a minimum link path, and the 
third algorithm uses the edge-edge visibility graph of $P$. The first two
algorithms are for simple polygons, and they run in $O(n + k \log n)$ time,
where $k$ denotes the number of reflections in the constructed path. 
The third algorithm runs in $O(n^2)$ time and works for
polygons with or without holes. 
The number of reflections in the path produced by the third algorithm can be at
most three times that in an optimal diffuse reflection path.

\section{Computing an optimal constrained diffuse reflection path}
\label{sec:mcdrp}
\medskip
Let $SP(s,t)=(s,u_1,u_2,\ldots, u_m,t)$, where $u_1,u_2,\ldots, u_m$ are
vertices of $P$. We know that $SP(s,t)$ can be computed in linear
time~\cite{lp-net-84}. Since $cdrp(s,t)$ intersects $SP(s,t)$ 
only at the eaves, all other edges of $SP(s,t)$ can be treated as polygonal
edges. So, we assume without loss of generality, that 
$s$ and $t$ lie on $bd(P)$. We first present an
algorithm for computing a $mcdrp(s,t)$ for the special case where no edge of
$SP(s,t)$ is an eave (see Figure \ref{eaveno1}). Later we allow eaves in
$SP(s,t)$. 

\subsection{$SP(s,t)$ has no eave}
\label{ssec:noeave}

\subsubsection{Characterizing the path $cdrp(s,t)$}
\label{sssec:path-noeave-char}

\medskip

We know that $SP(s,t)$ is either right convex or left convex as there 
is no eave. Without loss of generality, we assume that $SP(s,t)$ makes a right
turn at every vertex of the path, while traversing from $s$ to $t$. 
So, vertices of
$SP(s,t)$ belong to the counterclockwise boundary of $P$ from $s$ to $t$
(denoted as $bd_{cc}(s,t)$). Since a $cdrp(s,t)$ does not intersect any edge of
$SP(s,t)$, all turning points of the $cdrp(s,t)$ lie on the clockwise boundary
of $P$ from $s$ to $t$ (denoted as $bd_c(s,t)$). We have the following lemma
that establishes the convexity property of $cdrp(s,t)$, as we have for $SP(s,t)$
and $mlp(s,t)$.
\begin{lem}
Every $cdrp(s,t)$ is a convex and simple path.
\label{lemma:cdrpconvexsimple}
\end{lem}
\begin{proof}
Consider a $cdrp(s,t)=(s,z_1,\ldots,z_p,t)$ where $z_i$, 
$1 \leq i \leq p$, are the
turning points of the $cdrp(s,t)$ on $bd(P)$. 
Since the path is simple by definition, the next turning point
$z_{i+1}$ of $z_i$ cannot belong to $bd_c(s,z_i)$, for all $i$.
Further, since $cdrp(s,t)$ does not intersect any edge of $SP(s,t)$, 
$z_{i+1}$ must belong to $bd_c(z_i,t)$. 
Since each turning point $z_i$ of the $cdrp(s,t)$ is an
interior point of an edge of $bd_c(s,t)$,
the $cdrp(s,t)$ makes a right turn at $z_i$, for all $i$. 
Therefore, $cdrp(s,t)$ is a simple and convex path.
\end{proof}
\begin{corr}
The turning points $z_1, z_2,\ldots, z_p$ of 
$cdrp(s,t)=(s=z_0,z_1,\ldots,...,z_p,z_{p+1}=t)$ appear in clockwise order
along $bd_c(s,t)$. 
\label{corollary:sequence}
\end{corr}

For any point $p$ inside $P$, we say that the line segment 
$pu_i$ is a {\it left (or, right) tangent} from $p$ to $SP(s,t)$
at the vertex $u_i$ (see Figure \ref{eaveno1}), 
if both $u_{i-1}$ as well as $u_{i+1}$,
lie to the left (respectively, right) of the ray emanating from
$p$ through $u_i$. Note that $u_{i-1}$, $u_i$ and $u_{i+1}$ are 
consecutive vertices of $SP(s,t)$. We have the
following lemma.

\begin{lem}
The right and left tangents from any turning point $z_i$ of
a $cdrp(s,t)$ to $SP(s,t)$ lie entirely inside the simple polygon $P$.
\label{lemma:tangents}
\end{lem}

\begin{proof}
The proof follows from Lemma \ref{lemma:cdrpconvexsimple} due to the convexity
and simplicity of the $cdrp(s,t)$.
\remove{
}
\end{proof}

\begin{figure*}[ht]
\begin{minipage}[b]{1.0\textwidth}
\centering
\includegraphics[width=0.6\columnwidth]{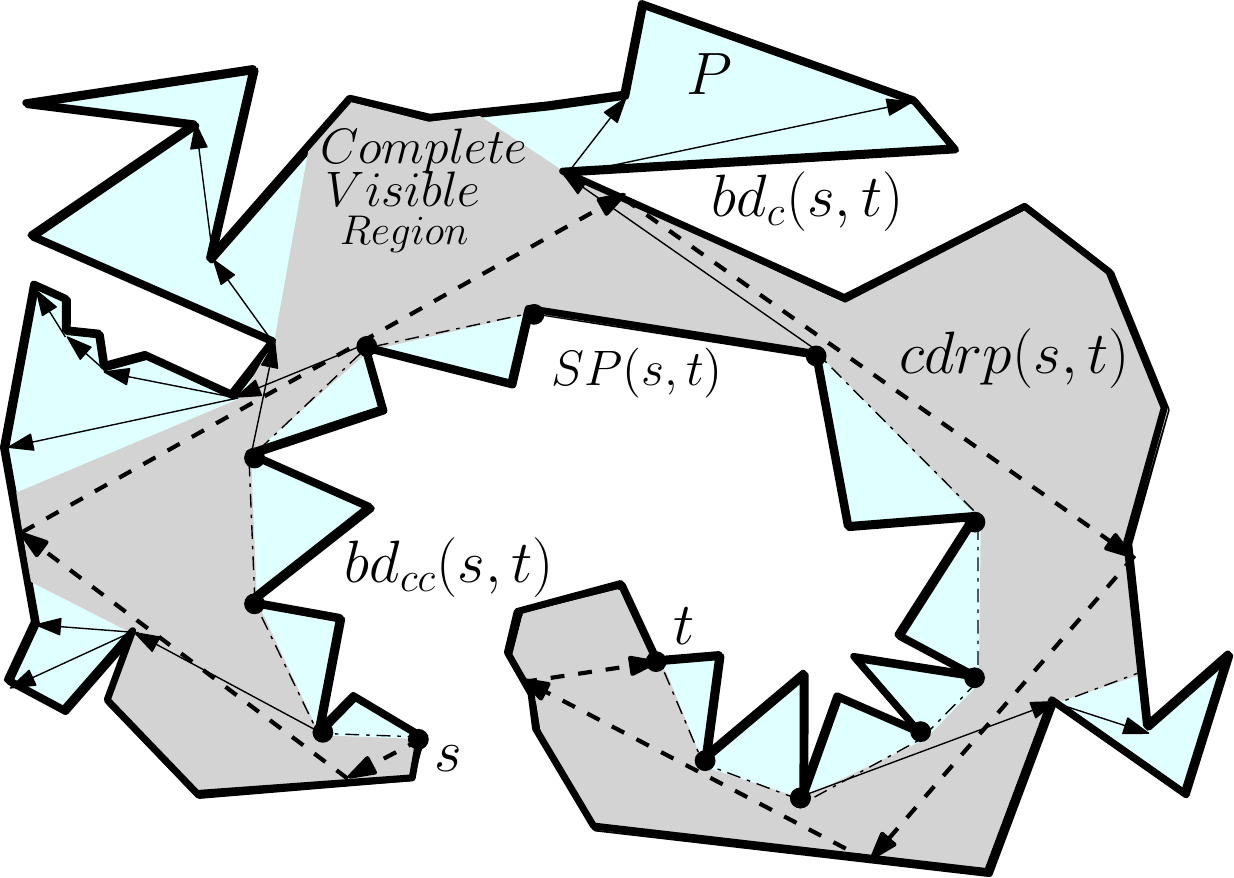}\\
\captionof{figure}{Every $cdrp(s,t)$ lies inside $CV(s,t)$.}
\label{eaveno2}
\end{minipage}
\end{figure*}

Let $CV(s,t)$ be the {\it complete visible} region of $P$ bounded by 
$SP(s,t)$ and $bd_c(s,t)$ such that the right and left tangents from 
every point $z$ of $CV(s,t)$ to $SP(s,t)$, lie inside $P$ (see Figure
\ref{eaveno2}). It follows from Lemma \ref{lemma:tangents} that a 
$cdrp(s,t)$ lies totally inside $CV(s,t)$ with turning points on the 
polygonal edges belonging to $CV(s,t)$ as stated in the following lemma.
\begin{lem}
Every $cdrp(s,t)$ lies entirely inside $CV(s,t)$.
\label{lemma:cvcdrp}
\end{lem}

\begin{figure*}[ht]
 \begin{minipage}[b]{1.0\textwidth}
  \centering
 \includegraphics[width=0.6\columnwidth]{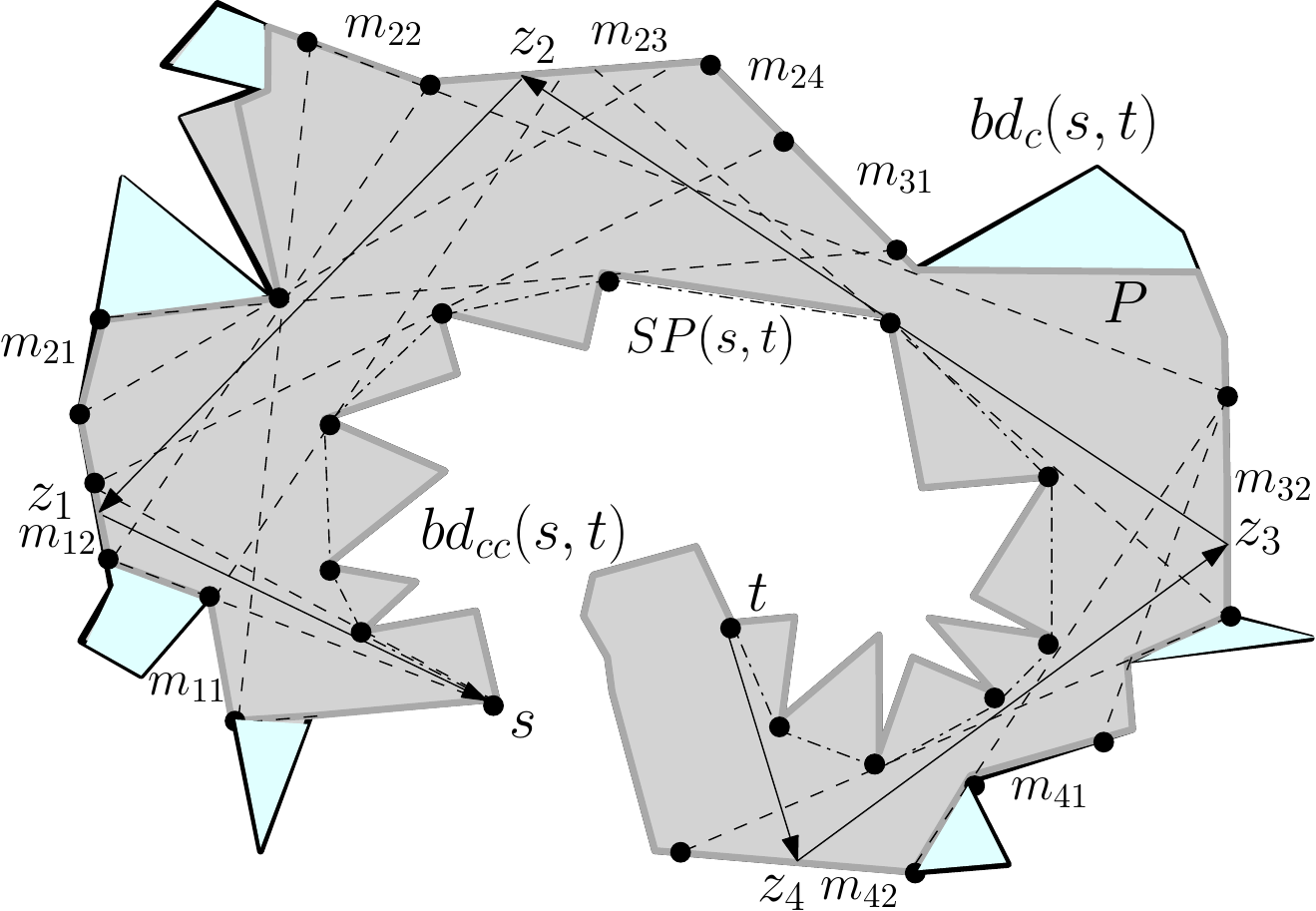}\\
\captionof{figure}{Mirrors $M_1=(m_{11},m_{12})$, $M_2=(m_{21},m_{22},m_{23},m_{24})$,
$M_3=(m_{31},m_{32})$, $M_4=(m_{41},m_{42})$ are computed in clockwise order along
$BCV(s,t)$.}
\label{lmirror1}
\end{minipage}
\end{figure*}
The region $CV(s,t)$ can be computed in linear time by traversing 
the shortest path trees inside $P$ rooted at $s$ and $t$ as given in 
\cite{cgmrs-ncml-95,g-cvpcs-91,ghosh-book-2007}. A {\it shortest path tree} 
rooted at a vertex $v$ is the union of the Euclidean shortest paths from 
$v$ to all vertices of $P$. The main property used by the algorithm in
traversing the trees is given below.

\begin{lem}\cite{cgmrs-ncml-95,g-cvpcs-91}
Let $u$ and $w$ be the parents of a vertex $v \in bd_{c}(s,t)$
in the shortest path trees inside $P$ rooted at $s$ and $t$. The vertex $v$
belongs to $CV(s,t)$ if and only if both $u$ and $w$ belong to
$SP(s,t)$\remove{(see Figure~\ref{eaveno2})}. 
\label{lemma:sptst}
\end{lem}

\subsubsection{Computing the reflecting edges for $cdrp(s,t)$}
\label{sssec:path-noeave-mirror}
The above discussion shows that all the turning points of a $cdrp(s,t)$
must lie on edges of $bd_c(s,t)$ that also belong to $CV(s,t)$. 
We denote the sequence of such intervals on polygonal edges 
as $BCV(s,t)$ $=$ $bd(P) \cap CV(s,t)$. We refer to polygonal edges 
containing edges of $BCV(s,t)$ as {\it reflecting} edges. Let us first identify
intervals of $BCV(s,t)$ on reflecting edges that can have 
the first turning point of a $cdrp(s,t)$. Let $M_1=(m_{11},m_{12},\ldots)$ 
be the intervals visible from $s$ on reflecting edges in clockwise order along
$bd_c(s,t)$ (see Figure \ref{lmirror1}). We refer to these intervals as 
{\it mirrors} of $M_1$. 
If $t$ is visible from any point $z_1$ on a mirror of $M_1$, then 
$(s,z_1,t)$ is a $cdrp(s,t)$.

\medskip

Assume that $t$ is not visible from any mirror of $M_1$. We identify 
intervals on reflecting edges such that every point in any such 
interval is visible from some point in a mirror of $M_1$.  
Note that a point on a reflecting edge may be visible from points in 
two or more mirrors of $M_1$. For each reflecting edge, the union of 
such intervals gives disjoint intervals on that reflecting edge. Let
$M_2=(m_{21},m_{22}, \ldots)$ be the intervals or mirrors in
clockwise order along $bd_c(s,t)$. Likewise, let $M_i=(m_{i1},m_{i2}, \ldots)$
be the mirrors created from the set $M_{i-1}=(m_{(i-1),1},m_{(i-1),2}, \ldots)$
of mirrors for $i\geq 3$. We have the following lemmas.

\medskip

\begin{lem}
All mirrors  of $M_2$ appear after all mirrors of $M_1$ in clockwise order 
along $BCV(s,t)$. 
\label{lemma:no-121}
\end{lem}
\begin{proof}
Let $a$, $b$ and $c$ be three points on $bd_c(s,t)$ in clockwise order
such that $a$ and $c$ are visible from $s$ but $b$ is not. In other words, $a$
and $c$  belong to two mirrors of $M_1$ but $b$ does not
belong to any mirror of $M_1$. Assume that $b$ is visible from $a$. 
So, no part of $bd_c(s,t)$ between $a$ and $b$ can intersect $bs$. 
If $bd_c(s,t)$ between $b$ and $c$ intersects $bs$, then $b$ does not belong to
$BCV(s,t)$. If $SP(s,t)$ intersects $bs$, then it also intersects
$cs$, contradicting the assumption that $c$ is visible from $s$. 
So, the second turning point of a $cdrp(s,t)$ must be on subsequent reflecting
edges of $M_1$ in clockwise order. 
\end{proof}
\begin{figure*}[ht]
\begin{minipage}[b]{1.0\textwidth}
\centering
\includegraphics[width=0.7\columnwidth]{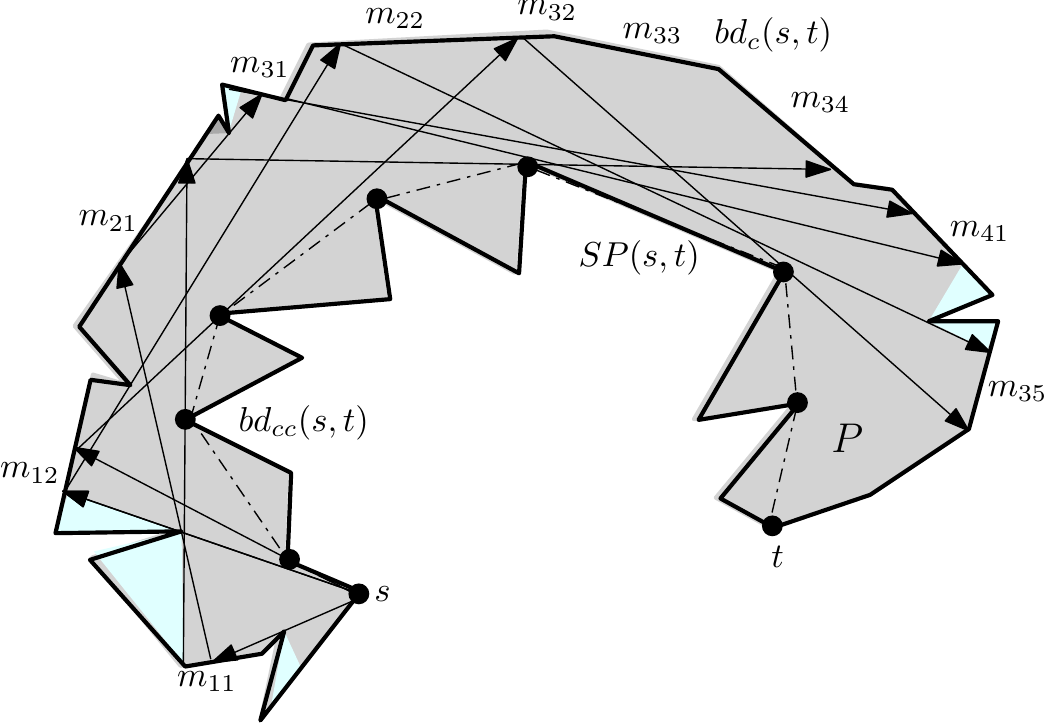}\\
\captionof{figure}{Mirrors of $M_2$, $M_3$ and $M_4$ are interleaved.}
\label{loverlapping}
\end{minipage}
\end{figure*}

The above relation between mirrors $M_1$ and $M_2$ cannot be
generalized for mirrors of $M_2,M_3,\ldots$, 
since mirrors in these sets may be interleaved as 
shown in Figure \ref{loverlapping}.
On the other hand, we have the following properties relating mirrors of
$M_1,M_2,M_3,\ldots$ and turning points of any $cdrp(s,t)$. 
\begin{lem}
No point on a mirror of $M_i$ is visible from any point on any mirror of
$M_1, \ldots, M_{i-2}$, for all $i \geq 3$.
\label{lemma:Mi-invisibility}
\end{lem}
\begin{proof}
 The proof follows from the definition of the mirrors $M_i$. 
\end{proof}
\begin{corr}
If every turning point on a $cdrp(s,t)$ belongs to a distinct $M_i$, then this
$cdrp(s,t)$ is an $mcdrp(s,t)$. 
\label{corollary:alternateinvisibility}
\end{corr}
\begin{proof}
If each of the $k$ turning points of a $cdrp(s,t)$ is from a distinct $M_i$,
then due to Lemma~\ref{lemma:Mi-invisibility}, the turning points must be on
mirrors of $M_1$, $M_2$, ..., $M_k$, respectively. Therefore, the $cdrp(s,t)$
is also an $mcdrp(s,t)$.
\end{proof}

\begin{figure*}[ht]
  \centering
\begin{minipage}[b]{1.0\textwidth}
\centering
\includegraphics[width=0.6\columnwidth]{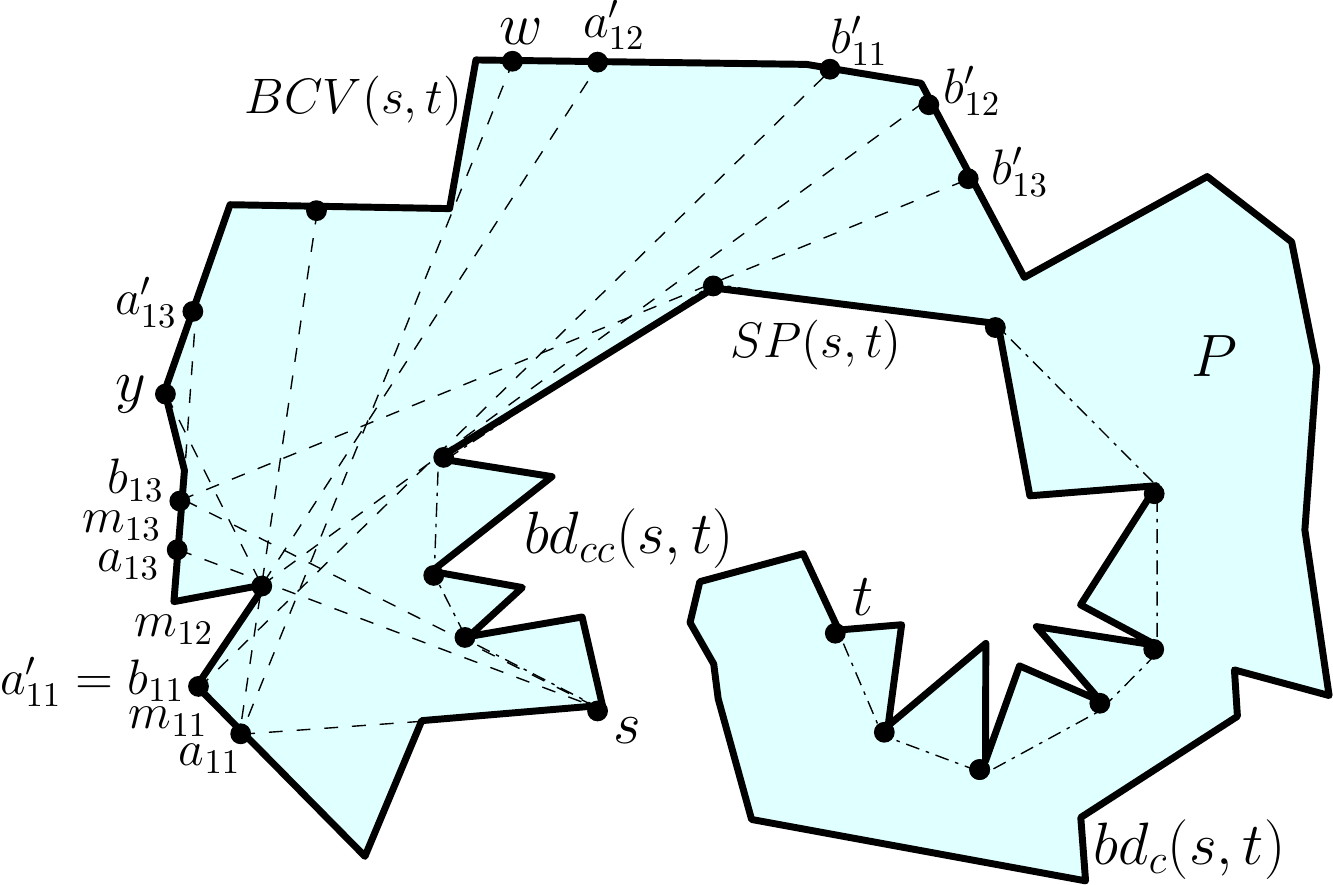}\\
\captionof{figure}{After locating $a'_{11}, a'_{12}, \ldots$ and 
$b'_{11}, b'_{12}, \ldots$ on $BCV(s,t)$, 
$M_2=(m_{21},m_{22},m_{23}, \ldots)$ is computed from 
$M_1=(m_{11},m_{12},m_{13},\ldots)$. Note that $a_{12}$, $a'_{11}$ and 
$b_{11}$ happen to be the same point in this figure.}
\label{lmirror-disjointness}
\end{minipage}
\end{figure*}

We next discuss the computation of endpoints of mirrors of $M_i$ from 
the mirrors of $M_{i-1}$ for $i\geq 2$. 
A point $x \in BCV(s,t)$ is said to be {\it weakly visible} from
a mirror if $x$ is visible from some point of the mirror.
Consider the first mirror $m_{11}$ (see Figure
\ref{lmirror-disjointness}). Let $a_{11}$ and $b_{11}$ be the endpoints 
of $m_{11}$, where  $b_{11}$ is the subsequent clockwise point of $a_{11}$ 
on $BCV(s,t)$. Let $a'_{11}$ be the first point in the clockwise order
on a reflecting edge of $M_2$ that is visible from $a_{11}$. 
Similarly, let $b'_{11}$ be the last point in the clockwise order
on a reflecting edge of $M_2$ that is visible from $b_{11}$. 
Observe that if the right tangent from $b_{11}$ to $SP(s,t)$ is extended
to $BCV(s,t)$, then it meets  $BCV(s,t)$ at $b'_{11}$. The portion of 
$bd_c(a'_{11},b'_{11})$ belonging to $BCV(s,t)$, and weakly 
visible from $m_{11}$ is called the {\it span} 
of $m_{11}$, and is denoted as $span(a'_{11},b'_{11})$. 
In the same way, the span of any mirror $m_{ij}$ can be defined. Observe that
$b'_{11}$, $b'_{12}$, $b'_{13}, \ldots$ occur in clockwise order along
$BCV(s,t)$.

\medskip

From the above definitions, the mirrors of $M_2$ formed due to reflections on
$m_{11}$ must belong to $span(a'_{11},b'_{11})$. 

\medskip

Let us identify mirrors of $M_2$ formed due to reflections on $m_{12}$.
These mirrors are formed on $span(a'_{12},b'_{12})$, that is,
on the portion of $bd_c(a'_{12},b'_{12})$ on $BCV(s,t)$, after
excluding the portion $span(a'_{11},b'_{11})$. 
If  $a'_{12} \in span(a'_{11},b'_{11})$, 
it follows that mirrors of $M_2$ formed due to the
reflection on $m_{12}$ must belong to the non-overlapping portion
$bd_c(b'_{11},b'_{12})$, 
weakly visible from $m_{12}$ and on $BCV(s,t)$. If 
$span(a'_{11},b'_{11})$ and $span(a'_{12},b'_{12})$ are disjoint, then 
mirrors of $M_2$ formed due to reflections on $m_{12}$
belong to $span(a'_{12},b'_{12})$. If
$span(a'_{12},b'_{12})$ 
contains $span(a'_{11},b'_{11})$, two weakly 
visible portions of $bd_c(a'_{12}, a'_{11})$ and $bd_c(b'_{11},b'_{12})$
on $BCV(s,t)$,
contain mirrors of
$M_2$  formed due to reflections on $m_{12}$.

\medskip

For identifying mirrors of $M_2$ formed due to reflections on $m_{13}$, remove
$span(a'_{11},b'_{11})$ and $span(a'_{12},b'_{12})$ from 
$span(a'_{13},b'_{13})$. Mirrors of $M_2$ formed due to reflections on $m_{13}$
lie on the remaining portions of  $span(a'_{13},b'_{13})$. 
Using this process of concatenations repeatedly, mirrors of $M_2$ can
be identified from the spans of mirrors of $M_{1}$. Furthermore, 
mirrors of $M_i$ can be identified given mirrors of $M_{i-1}$ in a similar
manner for $i\geq 3$. 

\medskip

Now we consider two kinds of mirrors possible on 
edges of $BCV(s,t)$. Some of these mirrors are {\it side}
mirrors, ending at vertices of $BCV(s,t)$. The others are {\it internal}
mirrors with endpoints in the interiors of edges of $BCV(s,t)$.
In the following lemmas, we bound the number of mirrors
using the maximum link distance (denoted as $\alpha (i)$) between
any two points on the first and last mirrors of $M_i$, for all $i$. 

\medskip 

\begin{figure*}[ht]
 \begin{minipage}[b]{1.0\textwidth}
  \centering
 \includegraphics[width=0.8\columnwidth]{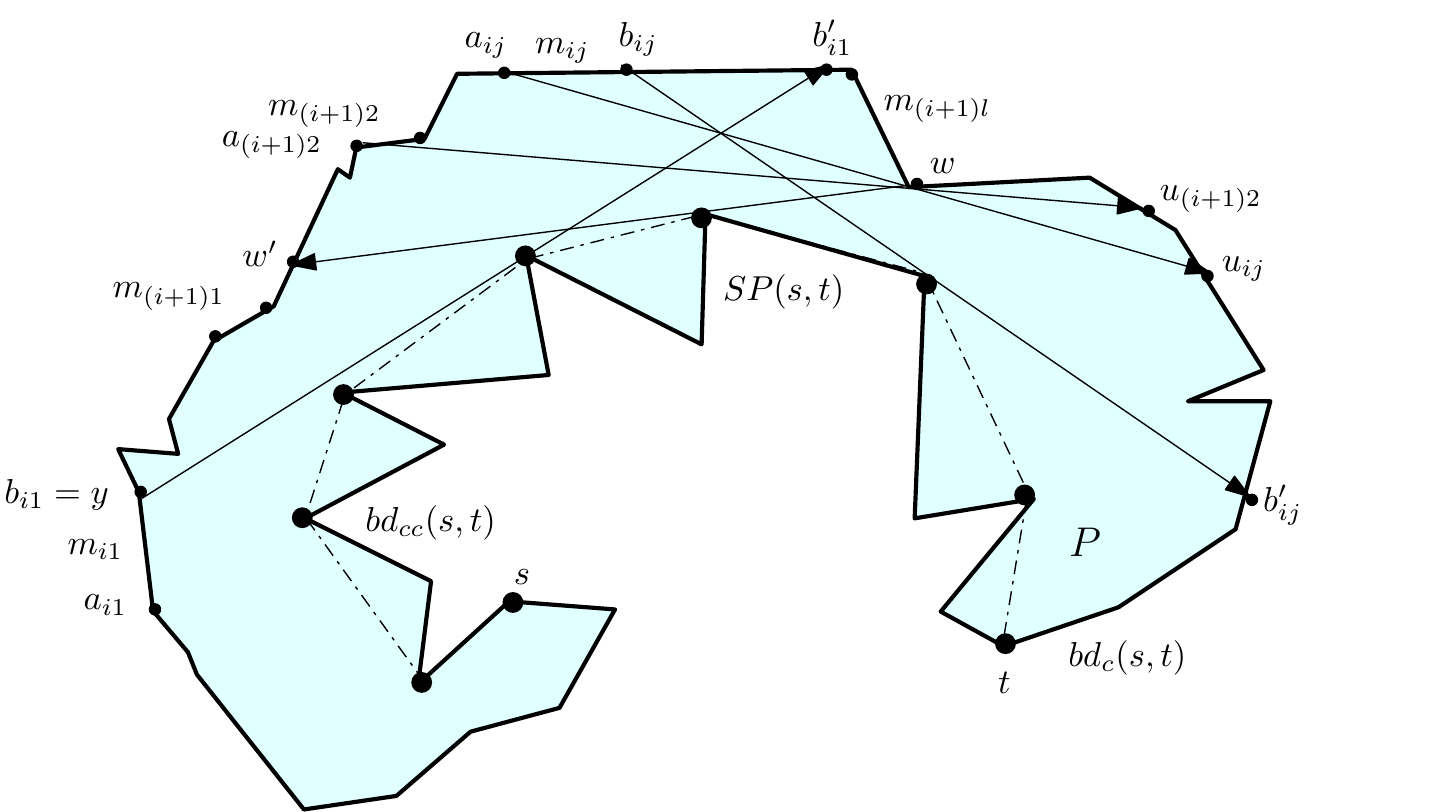}\\
\captionof{figure}{The vertex $w$ of $BCV(s,t)$ is visible 
from $m_{ij}$ and $m_{(i+1)2}$.}
\label{lconstructededge}
\end{minipage}
\end{figure*}

\begin{lem}
A vertex $w$ of $BCV(s,t)$ can be visible only from mirrors of $M_i$, 
$M_{i+1}$, $M_{i+2} \ldots$, $M_{i+\alpha(i)}$, where
$i$ is the smallest index such that a mirror of $M_i$ sees $w$.
\label{lthreeindex}
\end{lem}
\begin{proof}
Draw the left tangent from $w$ to $SP(s,t)$ and extend to $BCV(s,t)$ meeting it
at a point $w'$ (see Figure \ref{lconstructededge}). 
All mirrors that can see
$w$ must belong to $bd_c(w',w)$. Locate the mirror $m_{qj}$  on  
$bd_c(w',w)$ such that (i) $m_{qj}$ can see $w$, (ii) $q$ is the smallest among
all mirrors on $bd_c(w',w)$ that can see $w$, and
(iii) $j$ is the smallest among all mirrors of $M_q$ on $bd_c(w',w)$ 
that can see $w$. We set $i$ to $q$.
We know that any $mcdrp(s,w)$, can have only one turning point 
on a mirror amongst
all mirrors of $bd_c(w',w)$
due to Lemma \ref{lemma:Mi-invisibility}. We assume that at least one mirror of
$M_i$ does not see $w$, e.g., the first mirror $m_{i1}$ of $M_i$. This is the case
where $\alpha(i)=2$; the case where $\alpha(i)=1$ is simpler because 
mirrors of $M_{i+1}$ in $bd_c(w',w)$ may see $w$ but no mirror of
$M_k$,  for $k>i+1$ in the same region can see $w$ due to Lemma  
\ref{lemma:Mi-invisibility}. 
Let $y$ be the next clockwise vertex of $bd_c(s,w')$ after $m_{i1}$.  
We know that all mirrors of $M_i$, $M_{i+1}, \ldots$
belong to $bd_c(y,t)$. Assume that $b_{i1}b'_{i1}$ intersects $b_{ij}b'_{ij}$.
So, mirrors of $M_{i+1}$, created by $m_{i1}$ or any subsequent mirror of
$M_i$ on $bd_c(w',w)$  can see $w$. 
Since $b_{i1}b'_{i1}$ intersects $b_{ij}b'_{ij}$,
there can be at most one turning point of a $mcdrp(s,w)$ after the turning point
on $m_{i1}$ on any mirror of $M_{i+1}$ belonging to $bd_c(y,w')$.
These mirrors of $M_{i+1}$ can create mirrors of $M_{i+2}$ that can see $w$.
So, mirrors of  $M_{i+1}$ and $M_{i+2}$ can also see $w$. 
The same argument shows that if the maximum link distance from
a point on $m_{i1}$ to $w'$ is three, then mirrors of $M_i$, $M_{i+1}$, $M_{i+2}$
and  $M_{i+3}$ can see $w$,
as the mirror index increases by one for every link distance.
Thus, $w$ can be visible only from mirrors
of $M_i$, $M_{i+1}$, $M_{i+2} \ldots$, $M_{i+\alpha(i)}$ for some  value of $i$.
\end{proof}

\begin{lem}
Assume that $t$ is visible from a mirror of $M_k$ but not visible from any mirror of
$M_1,M_2,\ldots,M_{k-1}$ for $k<n$.
Let $\beta$ denote the maximum amongst 
\{$\alpha(3) (k-3), \alpha(4) (k-4),\ldots,\alpha(k)(k-(k-1))$\}.
The total number of mirrors $N$ in $M_1,M_2,\ldots, M_{k+\alpha(k)}$
is at most $O(n(n+\beta))$.
\label{lmirror-size}
\end{lem}
\begin{proof}
Let $w\in BCV(s,t)$ be visible from a mirror of $M_i$ in $bd_c(s,w)$, 
but not visible from any 
mirror of $M_{i-1}$ (see Figure \ref{lconstructededge}). 
Let $m_{ij}$ be the
first mirror of $M_i$ in the clockwise order
that can see $w$. So, $w$ belongs to $span(a'_{ij}, b'_{ij})$. Note that
$w$ may be visible from a mirror $m_{ij'}$ for $j'>j$ but $w$ is considered
only 
in $span(a'_{ij}, b'_{ij})$, since spans of mirrors of $M_i$ are concatenated to
form mirrors of $M_{i+1}$ as stated earlier. 
Therefore, $w$ is one endpoint of a mirror of $M_{i+1}$ (say, $m_{(i+1)l}$)
formed on the counterclockwise edge
of $w$ on $BCV(s,t)$, i.e., $w=b_{(i+1)l}$. 
If the next clockwise vertex of $w$
on $BCV(s,t)$ is visible from $m_{ij}$, then $w$ is also $a_{(i+1)(l+1)}$.
So, $w$ initiates two endpoints $a'_{(i+1)(l+1)}$ and $b'_{(i+1)l}$ of mirrors
of $M_{i+2}$, which, in turn, creates endpoints of mirrors for $M_{i+3}$, and
so on. So, $w$ can initiate at most two sequences of at most $k-i$ internal mirrors for 
$i\leq k-1$. 
Moreover, there can be mirrors of $M_{k+1},M_{k+2},\ldots,M_{k+\alpha (k)}$
between the first and last mirrors of $M_k$ that can also see $w$
due to Lemma \ref{lthreeindex}. 
So, any vertex $w$ that is visible from mirrors of $M_{k+1}$ can initiate
at most two sequences of at most $\alpha(k)-1$ internal mirrors.
Similarly, any vertex $w$ that is visible from mirrors of $M_{k+2}$ can initiate
at most two sequences of at most $\alpha(k)-2$ internal mirrors.
Let $n'$ be the number of such visible vertices $w$.
So, the total number of internal mirrors 
created is at most $2n(k-i)+2n'(\alpha(k))\leq 4nk$, where $\alpha(k)\leq k$. 

\medskip

Consider the other situation where the next clockwise vertex of $w$ on  
$BCV(s,t)$ is not visible from $m_{ij}$. Scan $BCV(s,t)$ from $w$ in 
the clockwise order until the point (say, $u_{ij}$) is located such that
$a_{ij}$, $w$ and $u_{ij}$ are collinear. So, $u_{ij}$ becomes $a_{(i+1)(l+1)}$.
Again, $w$ initiates two endpoints of mirrors of $M_{i+1}$, which create two
endpoints $a'_{(i+1)(l+1)}$ and $b'_{(i+1)l}$ of mirrors of $M_{i+2}$, and so
on. So, for $i\leq k-1$, each such vertex $w$ can initiate at 
most two sequences of at most $k-i$ internal
mirrors. Therefore, a total of at
most $2n(k-i)\leq 2kn$ mirrors can be created in $M_1,M_2,\ldots, M_k$, in this manner.
Moreover, there can be
mirrors of $M_{k+1},M_{k+2},\ldots,M_{k+\alpha (k)}$
between the first and last mirrors of $M_k$, which can create at most
$2n'\alpha(k)$ internal mirrors as shown earlier.    
So, the total number of internal mirrors 
created is at most $2n(k-i)+2n'(\alpha(k))\leq 4nk$, where $\alpha(k)\leq k$. 

\medskip

Let us count additional mirrors of $M_{i+2}$ that may be created due to $w$ 
if $w$ is also visible from mirrors of  $M_{i+1}$. 
Let $m_{(i+1)q}$ be the first mirror of $M_{i+1}$ in the clockwise order that
can see $w$. 
Note that $m_{(i+1)q}$ may lie before or after $m_{ij}$ on $bd_c(s,w)$. 
If $m_{(i+1)q}$ sees both edges of $w$, then no additional 
internal mirror
of $M_{i+2}$ is created on these edges because 
mirrors of $M_{i+1}$ are already
present. 
Consider the other situation, where the next clockwise 
vertex of $w$ on $BCV(s,t)$ is not visible from $m_{(i+1)q}$. Locate the 
next visible point 
$u_{(i+1)q}$ of $w$, as before by scanning $BCV(s,t)$ from $w$. If $u_{(i+1)q}$ does not
belong to $bd_c(w,u_{ij})$, then no additional mirror of $M_{i+2}$ is created because
a mirror of $M_{i+1}$ is already created due to $m_{ij}\in M_i$.
However, if $u_{(i+1)q}$  belongs to  $bd_c(w,u_{ij})$, then additional mirrors
of $M_{i+2}$ are created on $bd_c(u_{(i+1)q},u_{ij})$.  

\medskip 

If this situation
happens repeatedly in this fashion due to mirrors of $M_i$, $M_{i+1} \ldots$ for
$w$, then $\alpha(i)$ additional endpoints of mirrors may be created for $w$ 
due to Lemma \ref{lthreeindex} in addition to visible vertices of
$bd_c(w,u_{ij})$. Each such additional endpoint may initiate a sequence of at most $k-i$
internal mirrors. So, $w$ can cause the creation of a total of 
at most $\alpha(i)(k-i)$ internal mirrors for $i\leq k-1$.
Moreover, any such vertex  $w$ can also 
cause the creation of a total of at most 
$\alpha (k) (\alpha (k)-1)$ internal mirrors due mirrors 
in $M_{k+1},M_{k+2},\cdots,M_{k+\alpha(k)}$.
Therefore, all such visible vertices $w$ occurring on 
spans of different mirrors of
$M_1$, $M_2 \ldots$, $M_{k+\alpha(k)}$  can lead to at 
most $n\beta +\alpha^2 (k)$ internal mirrors. 
Hence, $N$ is bounded by $4nk+n\beta+\alpha^2(k)$, 
which is  bounded by $O(n(n+\beta))$.
\end{proof}

\medskip 

From now onwards we assume that $t$ is visible from a mirror of 
$M_k$ but not visible from any mirror of $M_1,M_2,...,M_{k-1}$.
Let us explain how mirrors of $M_1,M_2,\ldots, M_k$ can be computed by
traversing $BCV(s,t)$ from $s$ to $t$ in clockwise order using the 
method stated
in the proof of Lemma \ref{lmirror-size}. 
We know that each edge of $BCV(s,t)$,
partially or totally visible from $s$, is a mirror of $M_1$. 
Then, $b'_{11},b'_{12}, \ldots$ are computed by scanning $BCV(s,t)$ in  
clockwise order. Point $a'_{11}$ is the first point of $BCV(s,t)$
after $b_{11}$ in clockwise order that is visible from $a_{11}$.
After locating $a'_{11}$, weakly visible portions of $span(a'_{11},b'_{11})$ 
from $m_{11}$ are computed. Then the weakly visible portions of
$span(a'_{12},b'_{12})$ from $m_{12}$ are computed, after excluding
$span(a'_{11},b'_{11})$. Repeating this process of computing spans for the
remaining mirrors of $M_1$, all mirrors of $M_2$ are computed.
Similarly, mirrors of $M_3$ can be computed from mirrors $m_{21}, m_{22},
\ldots$. Repeating this process, mirrors of $M_4,M_5,\ldots, M_k$ can also be
computed.

\medskip

Recall that mirrors of $M_2$, $M_3 \ldots$, $M_k$  may be interleaved 
(see Figure \ref{loverlapping}).  
This means that edges of $span(a'_{ij},b'_{ij})$ may contain mirrors of $M_q$ 
for $q \leq i$, which should be excluded during the construction of new mirrors
of $M_{i+1}$ as explained in the proof of Lemma \ref{lmirror-size}. In other
words, $m_{ij}$ must introduce mirrors of $M_{i+1}$ only on the portions of 
$span(a'_{ij},b'_{ij})$ that
are not already visible from mirrors of $M_1$, $M_2\ldots$, $M_{i-1}$. However,
an edge of $span(a'_{ij},b'_{ij})$ may be visited $\alpha(i)$ times during this 
computation of mirrors from $\alpha(i)$ subsequent stages due to Lemma 
\ref{lthreeindex}.

\medskip

Let us explain how weakly visible edges are computed from mirrors 
of $M_i$.
Scan $BCV(s,t)$ in clockwise order from $a'_{i1}$ to $b'_{i1}$ and compute
the portions that are weakly visible from $m_{i1}$. If $a'_{i2}$ belongs to
$span(a'_{i1}, b'_{i1})$, continue the scan from  $b'_{i1}$ to $b'_{i2}$,
and  compute the portions that are weakly visible from $m_{i2}$. 
If $a'_{i2}$ does not belong to $span(a'_{i1}, b'_{i1})$, scan from 
$a'_{i2}$ to $b'_{i2}$ and compute the portions that are weakly visible 
from $m_{i2}$. Repeating this process for the
remaining mirrors of $M_i$, all mirrors of $M_{i+1}$ are computed by scanning
once. While computing mirrors of $M_{i+2}$ from mirrors of $M_{i+1}$, an edge of
$BCV(s,t)$ may be traversed again. This repetition can occur $\alpha(i)$ times
for some $i$. Note that whenever an edge of $BCV(s,t)$ is traversed, endpoints
of new mirrors are introduced on the edge. Therefore, the total cost of
traversing spans of all mirrors in all stages is bounded by the total number of
mirrors $N$ which is at most $O(n(n+\beta))$ due to Lemma \ref{lmirror-size}.
Also, the shortest path trees rooted at every vertex of $P$ are required
while scanning $BCV(s,t)$; these trees 
can be computed in $O(n^2)$ time using the 
the algorithm of Hershberger \cite{Hershberger89} for computing the visibility 
graph of a simple polygon. 
Hence, the overall time complexity of the algorithm
is $O(n(n+\beta))$. Note that $\beta$ can be $\Theta (n^2)$ for highly
skewed and winding simple polygons.

\medskip
\subsubsection{Computing $mcdrp(s,t)$, the $cdrp(s,t)$ with the minimum
number of turns}
\label{sssec:path-noeave-mcdrp}
Let us state how $mcdrp(t,s)=(t,z_k, \ldots, z_2,z_1,s)$ can be computed from 
$M_k,M_{k-1},\ldots,M_1$ satisfying
Corollary \ref{corollary:alternateinvisibility} (see Figure \ref{lmirror1}).
Consider the computation of $z_k$. Scan $BCV(s,t)$ from $t$ to $s$ and compute
\remove{the shortest path from $t$ to each vertex $v$,} 
until a point $z_k$ on a mirror $vw$ of $M_k$ is found to be visible from $t$,
where $z_k$ is the intersection point of $vw$ and the ray drawn from $t$ through
the next vertex of $SP(t,v)$. 
Note that $SP(t,v)$ makes only right turns and $z_k$ is a point directly visible
from $t$. Starting from $z_k$, a similar procedure can be adopted to locate a
point $z_{k-1}$ on a mirror of $M_{k-1}$, directly visible from $z_k$. 
Repeating this process, all turning points of the $mcdrp(s,t)$ can be computed
in linear time. Thus, we have the following lemma.
\begin{lem} 
If $SP(s,t)$ does not have an eave, then an $mcdrp(s,t)$ 
can be computed in $O(n(n+\beta))$ time, where $\beta = \Theta(n^2)$.
\label{theorem:algorunningtime}
\end{lem}

\subsection{$SP(s,t)$ has one or more eaves}
\label{ssec:oneeave}

We initiate the discussion with the case when $SP(s,t)$ contains 
one eave and then generalize. Let 
$SP(s,t)=(s,u_1,u_2,\ldots, u_{q-1},u_q,\ldots,u_m,t)$ 
where $u_{q-1}u_q$ is the eave (see Figure \ref{leaveone1}).
Without loss of generality, we assume that $SP(s,t)$ makes a right turn 
(left turn) at 
every vertex of $SP(s,u_q)$ ($SP(u_{q-1},t)$), 
while traversing from $s$ to $t$. 
So, the vertices of $SP(s,u_{q-1})$ belong to $bd_{cc}(s,t)$, and
the vertices of $SP(u_q,t)$ belong to $bd_c(s,t)$.
Let $cdrp(s,t)=(s,z_1,\ldots,z_{r-1},z_r,\ldots,z_p,t)$, where $z_{r-1}z_r$
intersects $u_{q-1}u_q$. 
Note that there is only one such intersection as per the definition of $cdrp(s,t)$.  
We have the following lemma.

\begin{lem}
Every $cdrp(s,t)$ makes a right turn at $z_1,z_2,\ldots,z_{r-1}$ on 
$bd_c(s,t)$, and a left turn at
$z_r,z_{r+1},\ldots,z_p$ on $bd_{cc}(s,t)$.
\label{lemma:oneeaveturns}
\end{lem}
\begin{proof}
The proof follows along the lines of the proof of Lemma
\ref{lemma:cdrpconvexsimple}.
\end{proof}

\medskip
\begin{figure*}[ht]
\begin{minipage}[b]{1.0\textwidth}
\centering
\includegraphics[width=0.6\columnwidth]{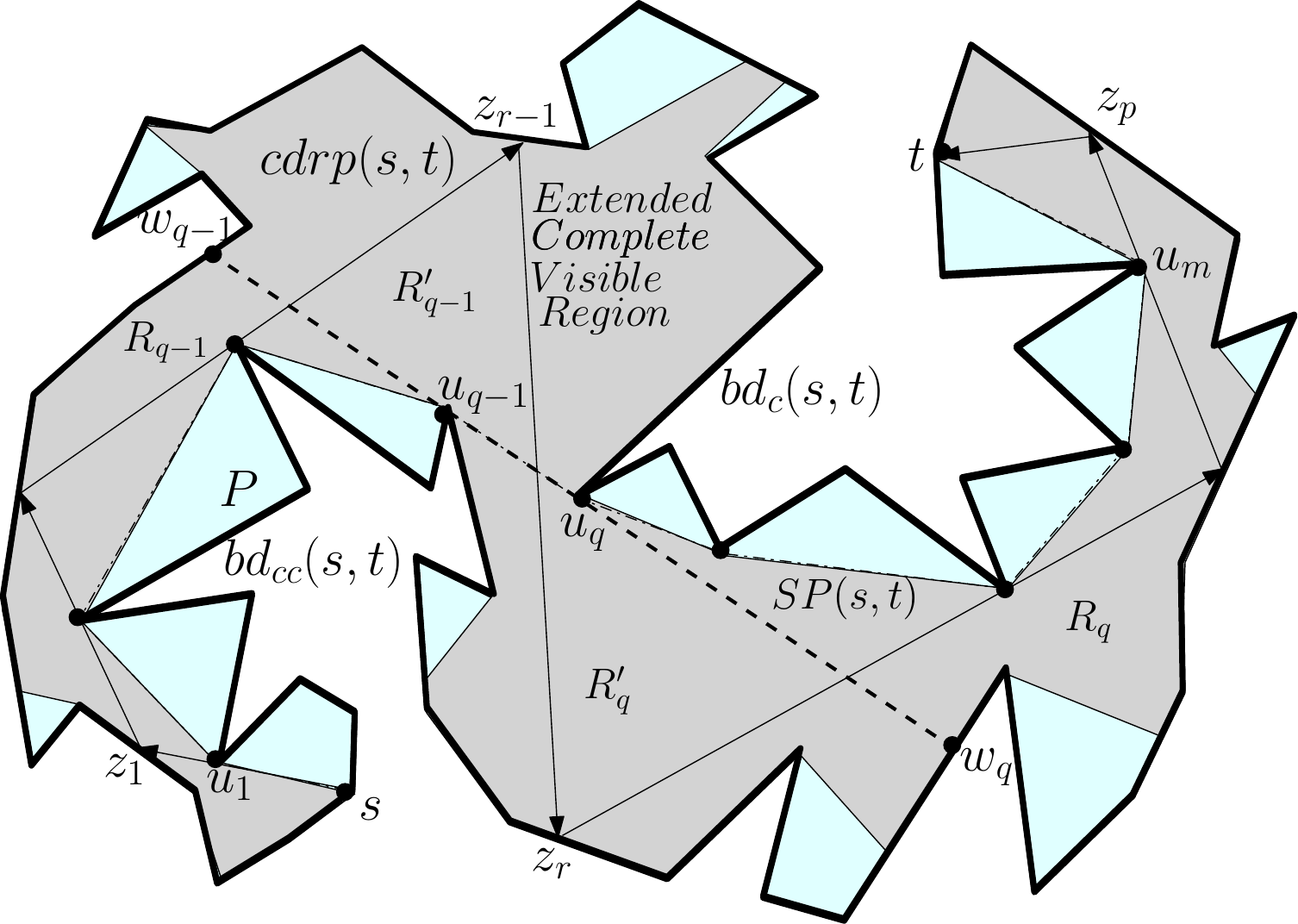}\\
\captionof{figure}{Every $cdrp(s,t)$ lies inside 
$ECV(s,t)$, \newline the extended complete visible 
region of $P$.}
\label{leaveone1}
\end{minipage}
\end{figure*}
In the case where $SP(s,t)$ had no eaves, we had computed an $mcdrp(s,t)$ based
on the analysis of the formation of mirrors on $BCV(s,t) = CV(s,t) \cap bd(P)$. 
If $SP(s,t)$ has eaves then we need to 
ensure that the $cdrp(s,t)$ crosses each eave exactly once. So, we 
need to consider the weak visibilty region of each eave and consider its
intersection with $CV(s,t)$, yielding the
{\it extended complete visibility region $ECV(s,t)$}. 
Extend eave $u_{q-1}u_q$ from $u_{q-1}$ to $bd_c(s,u_q)$, meeting it
at a point $w_{q-1}$ (see Figure \ref{leaveone1}). 
Similarly, extend $u_{q-1}u_q$ from $u_q$ 
to $bd_{cc}(u_{q-1},t)$, meeting it at a point $w_q$. 
Since $z_{r-1}z_r$ intersects $u_{q-1}u_q$, (i)
$z_{r-1}$ and $z_r$ must be weakly visible from $u_{q-1}u_q$, (ii) 
$z_{r-1}$ belongs to $bd_c(w_{q-1},u_q)$, and (iii)
$z_{r}$ belongs to $bd_{cc}(u_{q-1},w_q)$. Let $R_{q-1}$ be the region of $P$
bounded by $bd_c(s,w_{q-1})$, $SP(s,u_{q-1})$ and $u_{q-1}w_{q-1}$.  
Let $R'_{q-1}$ be the region of $P$ bounded by $bd_c(w_{q-1},u_q)$, and 
$w_{q-1}u_q$. Let $R'_{q}$ be the region of $P$ bounded by 
$bd_{cc}(u_{q-1},w_q)$, and $u_{q-1}w_q$. Let $R_{q}$ be the region of $P$
bounded by $SP(u_{q},t)$, $u_{q}w_{q}$, and $bd_{cc}(w_{q},t)$. 
We define
the {\it extended complete visible region $ECV(s,t)$}
as the set of points in $R_{q-1}\cup R'_{q-1}\cup R'_{q}\cup R_q$ such that 
(i) the left and right tangents from every point of $R_{q-1}$ to $SP(s,u_{q-1})$
lie inside $P$,
(ii) the left tangent from every point $z$ of $R'_{q-1}$ to $SP(s,u_{q-1})$ lies
inside $P$, and $z$ is visible from some point of $u_{q-1}u_q$,
(iii) the left tangent from every point $z$ of $R'_{q}$ to $SP(u_{q},t)$ lies
inside $P$, and $z$ is visible from some point of $u_{q-1}u_q$, or
(iv) the left and right tangents from every point of $R_{q}$ to $SP(u_{q},t)$
lie inside $P$. The shaded region in Figure~\ref{leaveone1} shows $ECV(s,t)$.
We have the following sequel to Lemma~\ref{lemma:cvcdrp}.
\begin{lem}
Every $cdrp(s,t)$ lies inside $ECV(s,t)$.
\label{lemma:insideecv}
\end{lem}
\begin{proof}
$CV(s,t)$ is augmented to form $ECV(s,t)$ as defined earlier. This 
augmentation considers the weak visibility region of an eave. As per the 
notations of Lemma~\ref{lemma:oneeaveturns}, the turning points 
$z_1, \ldots, z_{r-2}$ and $z_{r+1}, \ldots, z_p$ of $cdrp(s,t)$ lie 
in $CV(s,t)$ and hence in $ECV(s,t)$. The turning points $z_{r-1}$ and 
$z_r$ lie in the weak visibility region of the eave $u_{q-1}u_q$ and 
hence, in $ECV(s,t)$. Thus, the entire $cdrp(s,t)$ lies inside $ECV(s,t)$.
\end{proof}
\remove{
It follows from Lemma \ref{lemma:oneeaveturns} that every 
$cdrp(s,t)$ lies entirely inside $ECV(s,t)$ with turning points
on the polygonal edges of $ECV(s,t)$ as stated in the
following lemma.
}

\medskip
\begin{figure*}[ht]
 \begin{minipage}[b]{1.0\textwidth}
  \centering
 \includegraphics[width=0.6\columnwidth]{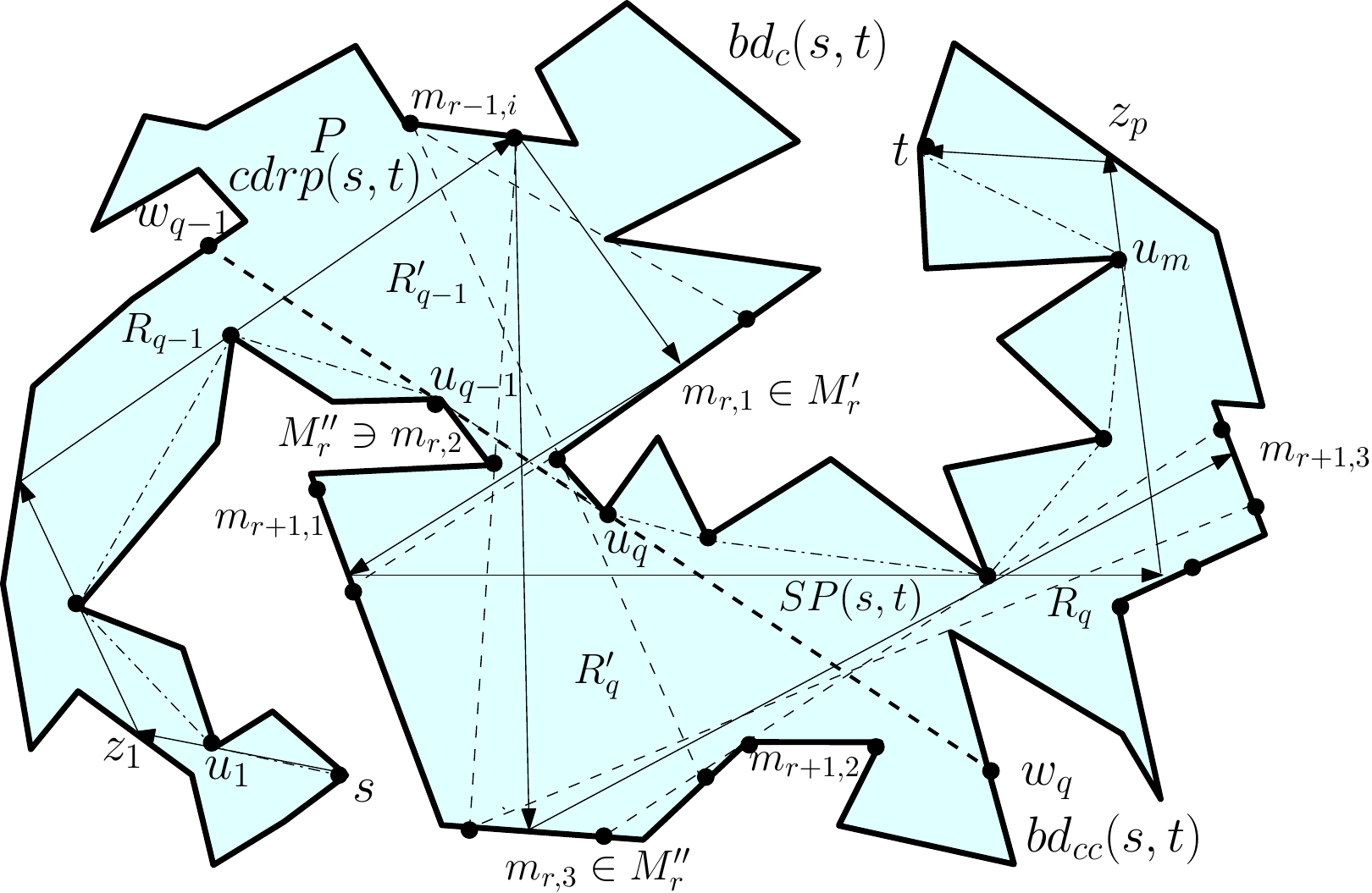}\\
\captionof{figure}{Starting from a mirror $m_{r-1,i}$, two mirror sets 
$M'_r=(m_{r,1})$ and $M''_r=(m_{r,2},m_{r,3})$ are formed on opposite sides of 
the eave $u_{q-1}u_q$. Note that $M_r=M'_r\cup M''_r=
(m_{r,1},m_{r,2},m_{r,3})$.}
\label{leaveone2}
\end{minipage}
\end{figure*}

We next proceed with the formation of mirrors on $ECV(s,t) \cap bd(P)$.
Let $m_{r-1,i}$ of $M_{r-1}$ be partially or entirely on a reflecting edge of
$R'_{q-1}$, where $r-1$ is the smallest index of all mirrors on reflecting edges of
$R'_{q-1}$ (see Figure \ref{leaveone2}). Since some points of $R'_q$ may become
visible from $m_{r-1,i}$, the next set of mirrors $M_r$ consists of two subsets
of mirrors $M'_r$ and $M''_r$, where all mirrors of $M'_r$ belong to reflecting
edges of $R'_{q-1}$, and all mirrors of $M''_r$ belong to reflecting edges of
$R'_q$. 
So, $M_r$ includes the union of mirror sets $M'_r$ and $M''_r$.
Observe that some mirrors of $M_{r+1}$ are created on the reflecting edges of
$R'_q$ due to mirrors in $M'_r$, and the mirrors of $M_{r+1}$ are 
created on the reflecting edges of $R'_q\cup R_q$ due to mirrors in $M''_r$. 
No mirror of $M_{r-1}$ can see any mirror of $M_{r+1}$,
satisfying Lemma \ref{lemma:Mi-invisibility}.

\medskip

For computing mirror sets $M'_r$ and $M''_r$, 
$M'_{r+1}$ and $M''_{r+1},\cdots$,
$M'_{r+\alpha (r-1)}$ and 
$M''_{r+\alpha (r-1)}$, the process can be viewed as first computing $M'_r,M'_{r+1},...,
M'_{r+\alpha (r-1)}$ on reflecting edges of $R'_{q-1}$, and then computing  
$M''_r,M''_{r+1},\cdots,
M''_{r+\alpha (r-1)}$ on reflecting edges of $R'_q$ 
from mirrors of reflecting edges of $R'_{q-1}$.
Observe that mirrors on reflecting edges of $R'_{q-1}$ can be created 
from mirrors of $R'_{q-1}$ of lower index 
(such as $M_{r-1,i}\in R'_{q-1}$), or from mirrors of $bd_c(s,w_{q-1})$.
So, the two types of mirrors on reflecting
edges of $R'_{q-1}$ can be constructed by scanning 
$bd_c(w_{q-1},u_q)$ twice. 
Computing mirrors of $M''_r,M''_{r+1},\cdots$ can be done by
scanning $bd_{cc}(u_{q-1},w_q)$ once. 

\medskip

For computing mirrors sets as mentioned above, we require to compute 
weakly visible regions from edges of $ECV(s,t)$. 
Shortest path trees rooted at vertices of $ECV(s,t)$ are 
computed following $bd_c(s,u_{q})$ 
using the method by Hershberger \cite{Hershberger89} as stated earlier.  
After computing $SPT(u_q)$, $SPT(u_{q-1})$
is computed treating the eave $u_{q-1}u_q$ as the next edge. 
Subsequently, shortest path trees are computed following 
$bd_{cc}(u_{q-1},t)$. 
%
%
The process of computing sets of mirrors
$M_{r+1}, M_{r+2}, \ldots , M_k$ continues as 
before, until $t$ becomes visible
from some mirror of $M_k$. 
Therefore,
the total number of mirrors created is clearly $O(n(n+\beta))$, 
as in Lemma \ref{lmirror-size}.
Finally, an $mcdrp(s,t)$ is computed as described in
Section~\ref{sssec:path-noeave-mcdrp}. 
We now summarize the result for polygons with one eave in the following lemma.

\begin{lem}
If $SP(s,t)$ has an eave $u_{q-1}u_q$, then an $mcdrp(s,t)$ can be computed in
$O(n(n+\beta))$ time.
\label{theorem:oneeave}
\end{lem}

%
%
%

%

\medskip

For the case of $SP(s,t)$ having two or more eaves, mirrors can be computed 
between every two consecutive eaves and across every eave, as explained earlier
until $t$ becomes visible from a mirror of $M_k$.
We have the following lemma.

\begin{lem}
If $SP(s,t)$ has two or more eaves, then an $mcdrp(s,t)$ can be computed in
$O(n(n+\beta))$ time.
\label{lemma:twoormoreeaves}
\end{lem}

\medskip

The major steps of the algorithm are 
stated as follows.

\begin{algorithm}[H]
\KwIn{A source $s$ and a destination $t$ inside an $n$-vertex simple 
polygon $P$}
\KwOut{$mcdrp(s,t)$, if it exists}

Compute the Euclidean shortest path $SP(s,t)$ from $s$ to $t$\;

Compute the extended complete visibility region $ECV(s,t)$ of $P$\;

Starting from $s$, compute mirrors of $M_1,M_2,\ldots, M_k$ until $t$
becomes visible from some mirror of $M_k$\;

Starting from $t$, compute $cdrp(t,s)=(t,z_k, \ldots, z_2,z_1,s)$ by
locating turning points on mirrors of  $M_k,M_{k-1},\ldots,M_1$\;

Output $cdrp(s,t)$ as the $mcdrp(s,t)$\;
\caption{Computing $mcdrp(s,t)$}
\label{algo:smhzl}
\end{algorithm}

We conclude the computation of the $mcdrp(s,t)$ with the following 
theorem. 
\begin{theo}
For a source point $s$ and a destination point $t$ inside an $n$-vertex 
simple polygon $P$, an $mcdrp(s,t)$ can be computed in $O(n(n+\beta))$ time
if a $cdrp(s,t)$ exists. 
\end{theo}
\begin{proof}
First of all note that, the 
specific way in which we construct mirrors as discussed in
Section~\ref{sssec:path-noeave-mirror} ensures that (a) the path is a simple
path, crossing $SP(s,t)$ at each of its eaves exactly once, and 
(b) we can always find a $cdrp(s,t)$, if one exists. 
The cost of computing the mirrors in $M_1,M_2,\cdots,M_k$ is $O(n(n+\beta))$
as shown in Lemmas \ref{lmirror-size}, 
\ref{theorem:algorunningtime}, \ref{theorem:oneeave} and \ref{lemma:twoormoreeaves}. 
%
%
%
The computed $cdrp(s,t)$ is actually also  
an $mcdrp(s,t)$ as its turning points are chosen on mirrors 
from mirror sets $M_i$ with minimum index $i$ on the reflecting edges,
as mentioned and required in
Corollary~\ref{corollary:alternateinvisibility}.

\end{proof}

\section{Exploring the relationship between $cdrp(s,t)$ and $drp(s,t)$}
\label{sec:explore}
In this section, we establish two properties relating (minimum) $cdrp(s,t)$
and (mimimum) $drp(s,t)$. The first one deals with an approximation ratio 
and the second one deals with a diameter. 
\subsection{Comparing the number of turns between an $mcdrp(s,t)$ and an
optimal $drp(s,t)$}
\label{ssec:approx}
Though not the main focus of our work, we explore the relation 
of an $mcdrp(s,t)$ with the optimal $drp(s,t)$ and the $mlp(s,t)$. 
We first consider 
the case where $SP(s,t)$ has no eaves. 
\begin{lem}
If $SP(s,t)$ does not have an eave, then the number of turns in 
an $mcdrp(s,t)$ is at most twice that of an optimal $drp(s,t)$.
\label{theorem:lb1}
\end{lem}
\begin{figure}[h]
\begin{center}
\center{\includegraphics[width=0.58\columnwidth]{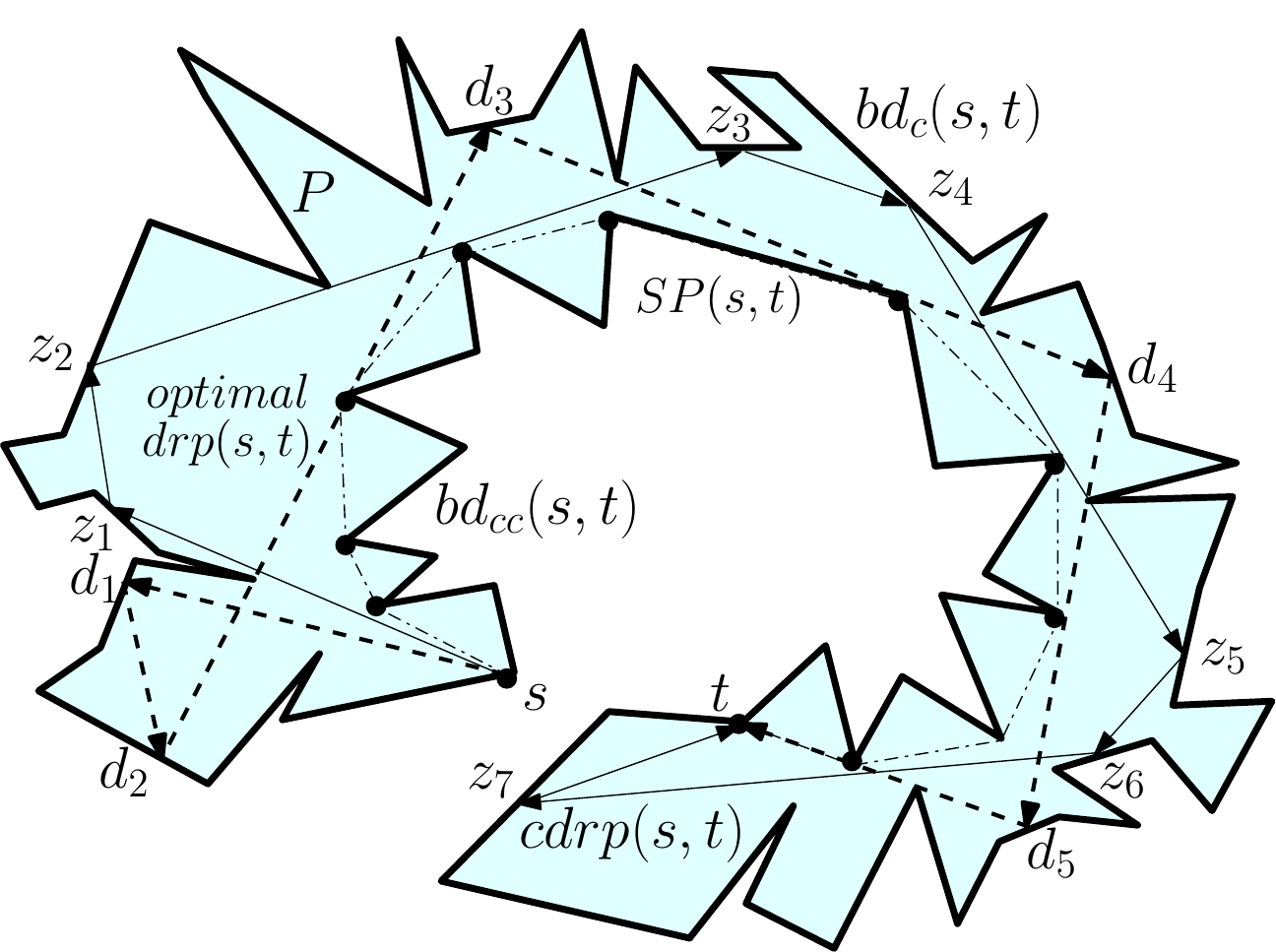}}
\caption{There may be two turning points in a $mcdrp(s,t)$ corresponding to
every turning point of an optimal $drp(s,t)$.}
\label{lmirror2}
\end{center}
\end{figure}
\begin{proof}
Let $d_i\in bd_c(s,t)$ be a turning point 
of an optimal $drp(s,t)$ (see Figure \ref{lmirror2}). Traverse this optimal path
from $d_i$ to $t$ until a turning point $d_j\in bd_c(d_i,t)$ is reached.
If $d_id_j$ is a segment and it does not intersect 
the $mcdrp(s,t)$, then there is no turning point 
of this $mcdrp(s,t)$ on $bd_c(s,t)$ between $d_i$ and $d_j$. 
If $d_id_j$ is a segment and it intersects
the $mcdrp(s,t)$, then there cannot be three or 
more turning points of   $mcdrp(s,t)$
on $bd_c(d_i,d_j)$ due to Lemma \ref{lemma:Mi-invisibility}, because in such 
a case the first and the third
turning points would become visible. 
If $d_id_j$ is not a segment, then no turning point of the 
optimal $drp(s,t)$ between $d_i$ and $d_j$ belongs to $bd_c(d_i,d_j)$. Then,
this path between $d_i$ and $d_j$ still intersects the
$mcdrp(s,t)$ exactly twice, and there can be at most two turning points
of the $mcdrp(s,t)$ on $bd_c(d_i,d_j)$. In the worst case,
all turning points of the optimal $drp(s,t)$ lie on $bd_c(s,t)$, 
and every link of this path intersects the $mcdrp(s,t)$ twice, and 
therefore, the number of turns in the $mcdrp(s,t)$ is at most twice 
that of the optimal $drp(s,t)$.
\end{proof}

\begin{lem}
If $SP(s,t)$ does not have an eave, then the number of links in 
an $mcdrp(s,t)$ is at most one less than twice the number of links in 
an $mlp(s,t)$.
\label{theorem:lb2}
\end{lem}

\begin{proof}
Consider the maximal chord $ab$ of any link of an $mlp(s,t)$, where 
$a,b\in bd(P)$ and $a\in bd_c(s,b)$. Due to Lemma \ref{lemma:Mi-invisibility}, 
there can be at most two 
turns of an $mcdrp(s,t)$ in $bd_c(a,b)$. For the first and the last link of
$mlp(s,t)$, there can be at most one link each, in an $mcdrp(s,t)$. So, the
number of links in an $mcdrp(s,t)$ is at most one less than twice the number of links in an $mlp(s,t)$.
\end{proof}

\medskip

\begin{figure}[h]
\begin{center}
\center{\includegraphics[width=0.8\columnwidth]{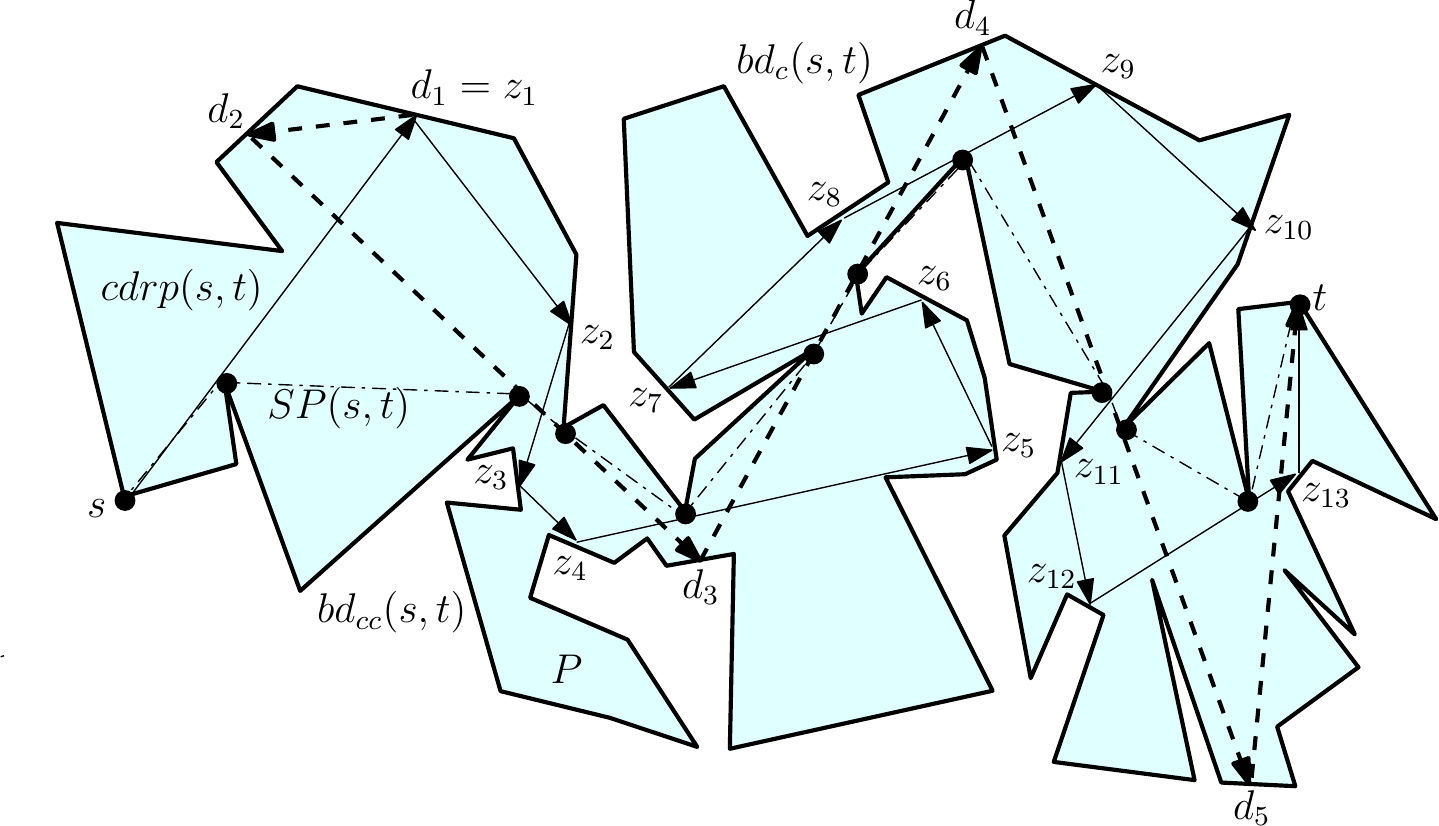}}
\caption{An optimal $drp(s,t)$ uses extensions of every eave of $SP(s,t)$ 
as its links. For every turn of the optimal $drp(s,t)$ after initial two turns, 
there are four turns of the $cdrp(s,t)$.  Observe that the
$cdrp(s,t)$ is also $mcdrp(s,t)$.}
\label{leaveopt2}
\end{center}
\end{figure}

We now generalize these results for any simple polygon $P$.
\begin{theo}
For any simple polygon $P$, the number of turns in any
$mcdrp(s,t)$ is at most $c \cdot opt$, where $2 \leq c \leq 4$.
Here $opt$ denotes the number of turns in an optimal $drp(s,t)$.
Moreover, the number of links in any $mcdrp(s,t)$ is at most 
$2((|mlp(s,t)|+l)-1)\leq 4|mlp(s,t)|-1$, where $|mlp(s,t)|$ is the number of links in an $mlp(s,t)$ and $l\leq |mlp(s,t)|$ is the number of eaves in $SP(s,t)$.
\label{theorem:approx-ratio}
\end{theo}
\begin{proof}
Let us count the number of turns $k$ taken by the $cdrp(s,t)$ computed by our
algorithm (Algorithm~\ref{algo:smhzl}). 
Recall the definition of $R_{q-1},R'_{q-1},R'_q,R_q$ as in Figure \ref{leaveone2}.
We know from Lemma \ref{theorem:lb1}
that there can be two turns in $R_{q-1}$ of the $cdrp(s,t)$ for every turn of an
optimal $drp(s,t)$. The same argument holds for turning points in $R_q$. If we
assume that an optimal $drp(s,t)$ uses $w_{q-1}w_q$ as a link in its path from
$s$ to $t$, then the $cdrp(s,t)$ can have two turning points in $R'_{q-1}$ and
two more turning points in $R_q$. So, $k \leq 2opt+2$ for one eave. If $SP(s,t)$
has $l$ eaves and an optimal $drp(s,t)$ uses extensions of every eave (see
Figure \ref{leaveopt2}), then $k \leq 2opt+2l$. Since $opt \geq l$, $k \leq
c\cdot opt$, where $2 \leq c \leq 4$.

Using Lemma \ref{theorem:lb2}, and considering the fact that at 
most two additional turns can be introduced near each
eave in an $mcdrp(s,t)$, it follows that the number of links 
in an $mcdrp(s,t)$ is at most  
$2(|mlp(s,t)|+l)-1)\leq 4|mlp(s,t)|-1$. 
\end{proof}

\medskip

Let us discuss how to cross-check the upper bound of $c$ for the $mcdrp(s,t)$
computed by our algorithm. It can be seen that counting the number of additional
turns actually taken by our computed $mcdrp(s,t)$ at each eave, a realistic
tighter upper bound for $c$ can be estimated using 
Theorem \ref{theorem:approx-ratio}. Moreover, the number of
turns in an optimal $drp(s,t)$ is at least the number of turns of a $mlp(s,t)$.
Since an $mlp(s,t)$ can
be computed in linear time \cite{g-cvpcs-91,ghosh-book-2007}, the entire
checking procedure can also be done in linear time.

In the worst case, the approximation ratio $c$ may be 4, and therefore 
worse than the result in \cite{ggmnps-acdr-2012}. Both the
approximation 
algorithms can be run and the one giving the minimum number of turns 
can be taken as the final result. However, it is interesting to observe 
that the $mcdrp(s,t)$ computed by our algorithm measures within a small 
constant factor of the optimal $drp(s,t)$.

\subsection{Constrained diffuse reflection diameter}
\label{ssection:diameter}

Let $a$ and $b$ be two vertices of $P$ such that the number of turns in an
optimal $drp(a,b)$ is maximum amongst all pairs of vertices of $P$.
The number of turning points in such a path is called the {\it diffuse
reflection diameter} of $P$. The relationship between the diffuse
reflection diameter of $P$ and the number of vertices $n$ of $P$ has been
studied in \cite{bcfhstw-corr-2013,kpabn2012}. Here we establish a similar
relationship for constrained diffuse reflection paths in $P$.
\begin{figure}[h]
\begin{center}
\center{\includegraphics[width=0.5\columnwidth]{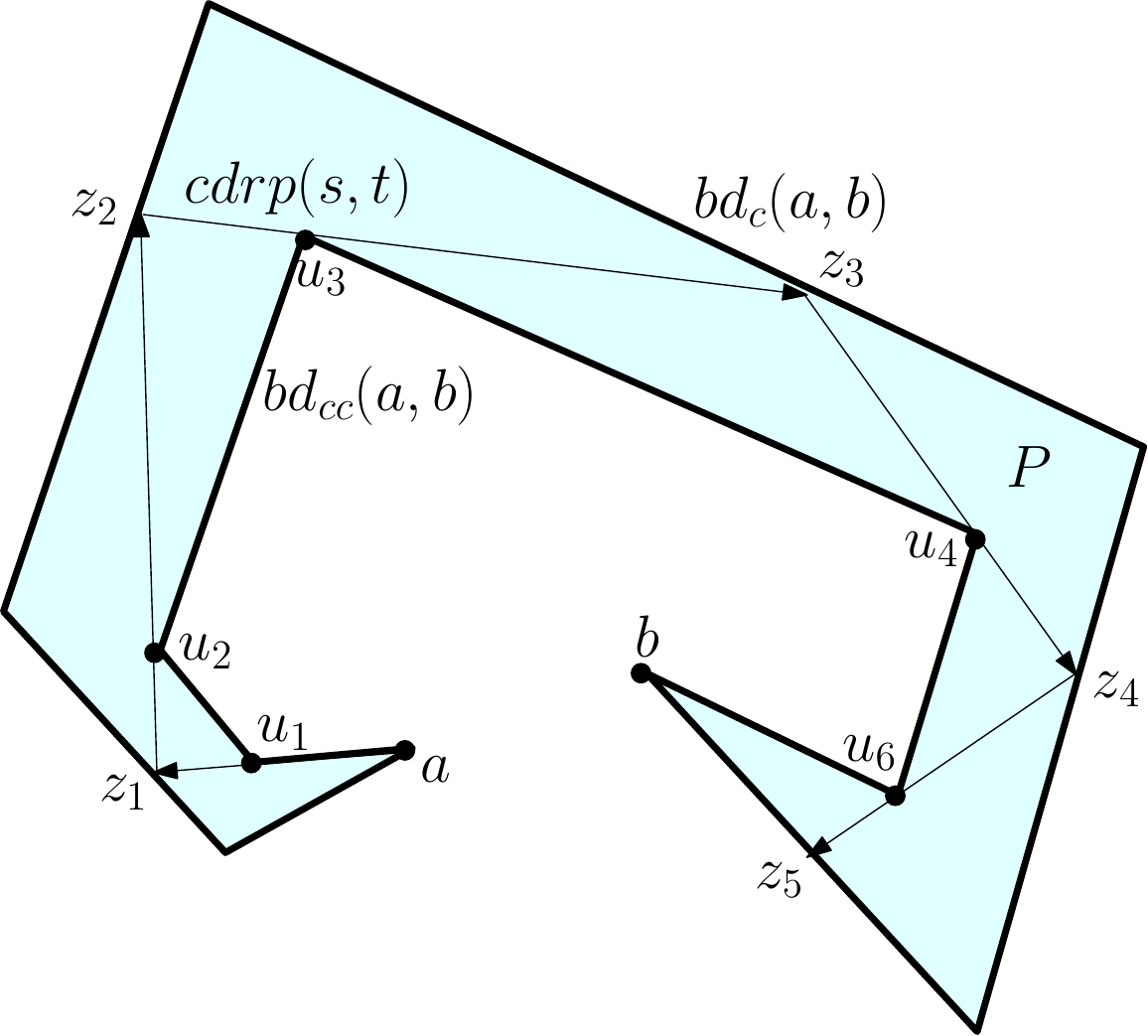}}
\caption{All edges of $bd_c(a,b)$ except the first one can have turning points
of a $mcdrp(a,b)$.}
\label{ldiameter}
\end{center}
\end{figure}

\begin{theo}
If there exists a $cdrp(a,b)$ in $P$, then the number of turns in
$mcdrp(a,b)<\frac n 2$. 
\label{lema:diameter}
\end{theo}
\begin{proof}
Let  $P$ be a spiral polygon of $n$ vertices (see Figure \ref{ldiameter}). 
Let $a$ and $b$ be two vertices of $P$ 
such that all vertices of $bd_c(a,b)$ are convex 
and all vertices of $bd_{cc}(a,b)$ are reflex.  
Since the diameter in $P$ is between $a$ and $b$, there always
exists a $cdrp(a,b)$ in $P$, with all turning points on $bd_c(a,b)$. 
Observe that every edge of $bd_c(a,b)$ except the edge incident on $a$ can have 
a turning point of $mcdrp(a,b)$. Since the non-consecutive turning points 
of the $mcdrp(a,b)$ are not mutually visible due to Lemma
\ref{lemma:Mi-invisibility}, 
and their visibility can only be blocked by $bd_{cc}(a,b)$, the number of edges
of $bd_{cc}(a,b)$ must be at least one more than
the number of turns of $mcdrp(a,b)$. So, the bound holds for
any spiral polygon $P$ and it is tight. 
This scenario is the worst-case for any simply polygon $P$, where $SP(a,b)$ has
no eaves.
 
\medskip

Assume that $SP(a,b)$ has one eave $u_{q-1}u_q$. Cut $P$ into two
sub-polygons using $u_{q-1}u_q$. Since the bound for spiral polygons holds for
each of these two sub-polygons, the bound also holds for $P$. If $SP(s,t)$ has
two or more eaves, then $P$ can be cut using each eave, and the bound holds for
the entire polygon $P$.
\end{proof}

Barequet et al. \cite{bcfhstw-corr-2013} have proved an upper bound of $\lfloor
{\frac n 2} \rfloor -1$ on the {\it diffuse reflection diameter} of  $P$. Though
a $drp(a,b)$ exists for any pair of 
points $a$ and $b$ in $P$, no $cdrp(a,b)$ may exist. Theorem~\ref{lema:diameter}
ensures that as long as a $cdrp(a,b)$ exists, the number of reflections in 
$mcdrp(a,b)$ has a similar worst-case upper bound of $\frac n 2$. This bound is
significant because such an $mcdrp(a,b)$ may have more turning points than an
optimal $drp(a,b)$.

\section{Concluding remarks}
\label{sec:conclude}
Our algorithm for computing an $mcdrp(s,t)$ in a simple polygon can be viewed
as a transformation of $SP(s,t)$ and $mlp(s,t)$ to $mcdrp(s,t)$. It will
be interesting to see whether an $mcdrp(s,t)$ can also be computed in a
polygon with holes using similar transformations. In such a scenario, 
observe that the region enclosed by $SP(s,t)$ and an $mcdrp(s,t)$ may contain
holes, making the problem difficult. 

\bibliographystyle{plain}
\bibliography{vis}

\begin{thebibliography}{10}

\bibitem{addpp-vmr-96}
B.~Aronov, A.~Davis, T.~Dey, S.~P. Pal, and D.~Prasad.
\newblock Visibility with multiple reflections.
\newblock {\em Discrete \& Computational Geometry}, 20:61--78, 1998.

\bibitem{addpp-vr-95}
B.~Aronov, A.~Davis, T.~Dey, S.~P. Pal, and D.~Prasad.
\newblock Visibility with one reflection.
\newblock {\em Discrete \& Computational Geometry}, 19:553--574, 1998.

\bibitem{adiy-cdfs-2006}
B.~Aronov, A.~R. Davis, J.~Iacono, and A.~S.~C. Yu.
\newblock The complexity of diffuse reflections in a simple polygon.
\newblock In {\em Proceedings of the 7th Latin American Symposium on
  Theoretical Informatics}, Lecture Notes in Computer Science, volume 3887,
  pages 93--104. Springer, Germany, 2006.

\bibitem{bcfhstw-corr-2013}
G.~Barequet, S.~Cannon, E.~Fox-Epstein, B.~Hescott, D.~L. Souvaine, C.~D.
  T{\'o}th, and A.~Winslow.
\newblock Diffuse reflections in simple polygons.
\newblock {\em CoRR}, abs/1302.2271, 2013.

\bibitem{cgmrs-ncml-95}
V.~Chandru, S.~K. Ghosh, A.~Maheshwari, V.~T. Rajan, and S.~Saluja.
\newblock N{C}-algorithms for minimum link path and related problems.
\newblock {\em Journal of Algorithms}, 19:173--203, 1995.

\bibitem{dr-vlt-1998}
A.~R. Davis.
\newblock {\em Visibility with Reflection in Triangulated Surfaces}.
\newblock PhD thesis, Polytechnic University, USA, 1998.

\bibitem{b-rsdo-1993}
M.~de\ Berg.
\newblock {\em Ray Shooting, Depth Orders and Hidden Surface Removal}.
\newblock Lecture Notes in Computer Science, vol. 703. Springer-Verlag, Berlin,
  Germany, 1993.

\bibitem{fdfhp-94}
J.~Foley, A.~van Dam, S.~Feiner, J.~Hughes, and R.~Phillips.
\newblock {\em Introduction to computer graphics}.
\newblock Addison-Wesley, Reading, MA, 1994.

\bibitem{g-cvpcs-91}
S.~K. Ghosh.
\newblock Computing visibility polygon from a convex set and related problems.
\newblock {\em Journal of Algorithms}, 12:75--95, 1991.

\bibitem{ghosh-book-2007}
S.~K. Ghosh.
\newblock {\em Visibility Algorithms in the Plane}.
\newblock Cambridge University Press, Cambridge, UK, 2007.

\bibitem{ggmnps-acdr-2012}
S.~K. Ghosh, P.~P. Goswami, A.~Maheshwari, S.~C. Nandy, S.~P. Pal, and
  S.~Sarvattomamda.
\newblock Algorithms for computing diffuse reflection paths in polygons.
\newblock {\em The Visual Computer}, 28(12):1229--1237, 2012.

\bibitem{Hershberger89}
J.~Hershberger.
\newblock Finding the visibility graph of a polygon in time proportional to its
  size.
\newblock {\em Algorithmica}, 4:141--155, 1989.

\bibitem{kpabn2012}
Arindam Khan, Sudebkumar~Prasant Pal, Mridul Aanjaneya, Arijit Bishnu, and
  Subhas~C. Nandy.
\newblock Diffuse reflection diameter and radius for convex-quadrilateralizable
  polygons.
\newblock {\em Discrete Applied Mathematics}, 161(10-11):1496--1505, 2013.

\bibitem{lp-net-84}
D.~T. Lee and F.~P. Preparata.
\newblock Euclidean shortest paths in the presence of rectilinear barriers.
\newblock {\em Networks}, 14:393--415, 1984.

\bibitem{pbs-lwlb-04}
S.~P. Pal, S.~Brahma, and D.~Sarkar.
\newblock A linear worst-case lower bound on the number of holes in regions
  visible due to multiple diffuse reflections.
\newblock {\em Journal of Geometry}, 81:5--14, 2004.

\bibitem{ppd-vmdr-98}
D.~Prasad, S.~P. Pal, and T.~Dey.
\newblock Visibility with multiple diffuse reflections.
\newblock {\em Computational Geometry: Theory and Applications}, 10:187--196,
  1998.

\bibitem{s-ltamlp-86}
S.~Suri.
\newblock A linear time algorithm for minimum link paths inside a simple
  polygon.
\newblock {\em Computer Graphics, Vision, and Image Processing}, 35:99--110,
  1986.

\end{thebibliography}
\end{document}